\documentclass[hidelinks,onefignum,onetabnum]{siamart250211} 

\usepackage{tikz}
\usepackage{tikz-3dplot}
\usepackage{pgfplots} 
\pgfplotsset{compat=1.18} 
\usepackage{graphicx} 
\newsiamremark{remark}{Remark}

\headers{Nonlocal Theory and SPH Simulation of Pancake Bouncing}{Z. Qiao, Z. Wang  and Y. Wei} 

\title{An Effective  Model for Droplet Impact Dynamics on Micro-Structured Surfaces: Nonlocal Theory and SPH Simulation of Pancake Bouncing}

\author{
ZHONGHUA QIAO \thanks{Department of Applied Mathematics, The Hong Kong Polytechnic University,  Hung Hom, Kowloon, Hong Kong
  (\email{zhonghua.qiao@polyu.edu.hk}).}
\and ZUANKAI WANG \thanks{Department of Mechanical Engineering, The Hong Kong Polytechnic University, Hung Hom, Kowloon, Hong Kong
  (\email{zk.wang@polyu.edu.hk}).}
\and YIFAN WEI \thanks{Corresponding author. Department of Applied Mathematics, The Hong Kong Polytechnic University,  Hung Hom, Kowloon, Hong Kong
  (\email{yi-fan.wei@polyu.edu.hk; weiyifan806@gmail.com}).}}  

\begin{document} 

\maketitle 

\begin{abstract}
The accurate mathematical modeling of droplet impact dynamics on micro-structured surfaces is fundamental to understanding and predicting complex fluid behaviors relevant to a wide range of engineering and scientific applications. In particular, the pancake bouncing phenomenon—systematically studied by Liu et al. (\textit{Nature Physics}, 2014)—on superhydrophobic micro-structured substrates presents significant theoretical challenges. Central to these challenges is the need to construct {effective} mathematical models that capture the intricate influence of substrate micro/nanostructures on droplet dynamics. This requires the development of robust formulations for surface tension, contact line dynamics, and the interaction forces between fluid and solid structures. In this work, we formulate a nonlocal mathematical framework for the simulation of 3D pancake bouncing on superhydrophobic micro-cone arrays. The model incorporates intermolecular attractive forces to represent droplet surface tension, and we provide a strict theoretical derivation linking these forces quantitatively to the macroscopic surface tension coefficient, thereby circumventing the reliance on empirical parameter tuning. The complex geometry of micro-cone arrays introduces fundamental difficulties in defining local normal directions for contact algorithms. To overcome this, we develop a nonlocal contact repulsion force model that governs fluid-solid interactions and ensures numerical stability under high Weber number conditions.
Based on this mathematical foundation, we implement the model using smoothed particle hydrodynamics (SPH), enabling high-precision 3D simulations. 
Computational experiments, validated against empirical data, confirm the model’s accuracy and robustness, while underscoring the key role of numerical simulation in elucidating droplet-microstructure interactions.
\end{abstract} 


\section{Introduction}
Understanding and controlling droplet rebound dynamics is essential for advancing liquid manipulation technologies in various engineering applications.
This phenomenon is particularly critical in applications such as combustion systems \cite{sirignano2014advances}, inkjet printing \cite{wijshoff2018drop}, anti-icing surfaces \cite{bird2013reducing}, corrosion-resistant coatings \cite{zhang2016superhydrophobic} and self-cleaning materials \cite{blossey2003self}.
Recent studies reveal that these technologies have broad applications across fields including aerospace, biomedicine, automotive engineering, and smart textiles \cite{khan2022recent}.
While hydrophobic coatings facilitate droplet removal, achieving directional control and efficient shedding requires not only low surface energy but also carefully designed nanostructures with tailored surface chemistry \cite{nakajima2011design}.
At such nano-enhanced interfaces, pancake bouncing \cite{liu2014pancake} emerges as a remarkable phenomenon: droplets impacting nanostructured superhydrophobic surfaces lift off immediately in a flattened state without retracting, reducing contact time by approximately 75\% compared to conventional rebound.
Precise liquid control is crucial for advanced functional surfaces and devices, necessitating deeper research into droplet rebound dynamics, especially given their significant implications in healthcare and electronics.

The dynamics of droplet rebound has garnered considerable scientific attention, driving a steady stream of key discoveries in interfacial science\cite{mohammad2023physics}. 
Early studies by \cite{schutzius2015spontaneous} revealed that droplets on superhydrophobic surfaces in low-pressure environments can self-remove through spontaneous levitation and trampoline-like bouncing.
 Later, \cite{hao2015superhydrophobic}  identified a superhydrophobic-like bouncing regime on thin liquid films, where the rebound dynamics were independent of the underlying substrate.
Subsequent work by \cite{liu2016superhydrophobic} introduced a facile method to fabricate hierarchical multilayered superhydrophobic surfaces for efficient droplet shedding.
\cite{yun2017bouncing} further explored symmetry-breaking impacts, showing that ellipsoidal droplets reduce contact time and suppress rebound magnitude.
Recent advances have significantly enhanced control over droplet rebound behavior through novel interfacial strategies.
\cite{zhao2023golden} proposed a 'golden section' design criterion to regulate rebound via structural spacing, challenging conventional approaches.
\cite{fang2023target} developed a leaf-inspired superhydrophobic cantilever for directional bouncing,
while \cite{lathia2023two} reduced contact time by coating droplets with hydrophobic particles.
Additionally, through the utilization of innovative tools, a deeper understanding of the underlying mechanisms has been achieved.
In their study, \cite{kumar2024heat}  utilized high-speed infrared thermography and a three-dimensional transient heat conduction COMSOL model to accurately map the dynamic heat flux distribution during droplet impact on a cold superhydrophobic surface.
\cite{goswami2023simultaneous} introduced a novel micro-controlled droplet generator capable of releasing two equally sized water droplets simultaneously on-demand, enabling the investigation of the behavior of two droplets on a dry substrate.
\cite{yue2024controllable} reported that controlled self-transport of bouncing droplets can be achieved with the assistance of ultra-smooth surfaces featuring wedge-shaped grooves.
\cite{qian2024emergence}  uncovered molecular-scale wetting dynamics through in situ characterization, providing fresh insights into droplet-surface interactions.
And \cite{sobac2025small} conducted a comprehensive numerically assisted analysis based on verifiable assumptions such as quasi-stationarities and small Reynolds/Peclet numbers.

Although experimental investigations provide critical observations, the numerical modeling is essential for achieving  fundamental physical understanding and enabling technological applications.
In fact, despite extensive experimental investigations into droplet impact phenomena, the fundamental physics governing collisions with hydrophobic surfaces remains inadequately understood, as noted in literature \cite{khojasteh2016droplet}.
On the contrary, numerical modeling offers unique insights into interfacial phenomena.
For example, \cite{huang2023phase} conducted numerical simulations based on the Cahn-Hilliard-Navier-Stokes model to achieve a complete and quantitative study on the effects of interfacial tension when it becomes sufficiently low.
Additionally, \cite{zhang2018numerical} simulated buoyancy-driven bubble-wall interactions using an adaptive ALE method, capturing the interaction of a rising bubble with a solid wall.
Furthermore, \cite{ray2024new} investigated the head-on collision of unequal-size droplets of the same liquid on wetting surfaces using direct numerical simulations at different size ratios.
Based on the Onsager variational principle, \cite{qian2024emergence} presented a phenomenological model that reveals how the intrinsic material parameters of soft gels dictate phase-separation dynamics.
Moreover, fabricating nanostructured substrates for droplet impact studies incurs significant costs, whereas numerical simulations provide a cost-efficient alternative for parametric studies and experimental design optimization.
Within this framework, numerical simulations have become pivotal in impact dynamics research. The creation of high-fidelity computational models capable of capturing intricate interfacial processes is essential for both theoretical advancements and practical applications.

Smoothed Particle Hydrodynamics (SPH) method, as a mesh-free Lagrangian method,  offers distinct advantages in simulating droplet rebound phenomena. In fact, droplet rebound typically involves large interfacial deformations and topological changes, presenting inherent difficulties for grid-based simulations.
Traditional Eulerian approaches, such as the Lattice Boltzmann method (LBM), rely on multiphase domain discretization and phase-field method to accurately capture fluid interfaces.
While LBM excels in modeling complex multi-phase flows, its application to gas-liquid systems with high density ratios demands non-trivial modifications to ensure numerical stability \cite{QIAO20191216,qiao2024free, AAMM-17-4}.
By contrast, SPH circumvents these limitations through its fully Lagrangian mesh-free formulation, eliminating the need for explicit gaseous-phase discretization altogether \cite{king2023large}.
This intrinsic capability renders SPH particularly suited for investigating droplet interactions with hydrophobic surfaces. Other alternative particle-based approaches such as molecular dynamics (MD) and dissipative particle dynamics (DPD) suffer from severe scalability limitations in mesoscale droplet simulations due to their demanding resolution requirements and consequent computational overhead.
The SPH method achieves an optimal balance between numerical resolution and computational efficiency. Notably, while SPH has seen extensive practical applications in recent years, its theoretical foundations have been progressively strengthened. Significant contributions include Du et al.'s mathematical analysis from a nonlocal perspective \cite{du2020mathematics,lee2019asymptotically} and Sun et al.'s stability analysis through an energy-based framework \cite{feng2023energy,AAMM-14-5,zhu2024energy}. These theoretical advancements have collectively reinforced SPH's position as a robust and well-validated methodology for droplet dynamics research.

Surface tension modeling constitutes a critical component of SPH simulation for droplet dynamics, and numerous significant studies have been conducted in this area.
Distinct from phase-field methods that regulate interface motion through energy minimization, SPH implementations necessitate the explicit definition of forces acting on interface particles to control interfacial behavior. In reality, researchers have developed two principal methodological approaches.
The first approach is the Continuous Surface Force (CSF) formulation \cite{brackbill1992continuum}, which implements a volumetric force at fluid interfaces through local curvature estimation.
While this method maintains physical consistency, it demands careful handling of surface normal vectors and curvature calculations near interface regions \cite{huo2024modeling,dong2023simulation,blank2024surface}. Furthermore, this approach faces challenges in simulating droplet topological changes.
The second approach involves the particle-particle interaction force (PIF) model \cite{tartakovsky2005modeling,breinlinger2013surface}, which bypasses the need for normal vector and curvature calculations.
This method demonstrates versatility in simulating solid-liquid interactions and solid wettability.
Despite its practical effectiveness, the second approach dependence on empirically determined strength of the interparticle attraction \cite{dong2024droplet}, lacking strong theoretical foundations, constitutes a major limitation. Notice that surface tension emerges from interparticle attractive forces. Thus, when these attraction strengths are empirically predetermined, the resulting surface tension coefficients become inherently uncontrollable, frequently leading to physical inconsistencies.  Therefore, developing a method that preserves theoretical consistency while ensuring conciseness and broad applicability has become critically important.

The contact algorithms is an another challenge for stable and precise SPH simulation, particularly for simulating droplet collisions on intricate arrays of microcones.
The inherent complexity of droplet impact dynamics, characterized by multiscale solid-liquid interactions, creates escalating challenges for contact algorithm stability as the Weber number increases \cite{pan2009binary}.
The traditional repulsive force algorithms \cite{huo2024modeling,dong2023simulation,dong2019modeling}, as a popular contact algorithm,  require accurate surface normals $\boldsymbol{n}$ and vertical distance $d_p$ to surface.
Specifically, the repulsive force $\boldsymbol{F}_i^{\text{rep}}$ on particle
$i$  satisfies the following relation:
\begin{align}\label{repulsive}
\boldsymbol{F}_i^{\text{rep}} \propto \frac{2|d_0-d_p|}{(\Delta t)^2} m_i\boldsymbol{n}_i,
\end{align}
where $\Delta t$ is time step, $m_i$ is particle mass and $d_0$ is contact threshold distance.
Although this approach demonstrates satisfactory performance for ideally smooth substrates \cite{mohammad2023physics}, its extension to realistically rough substrates encounters fundamental limitations. A representative case involves micro-structured hydrophobic surfaces \cite{liu2014pancake}, as illustrated in Fig. \ref{droplet_impact}.
\begin{figure}[!t]
\centering
\includegraphics[width=0.4\linewidth]{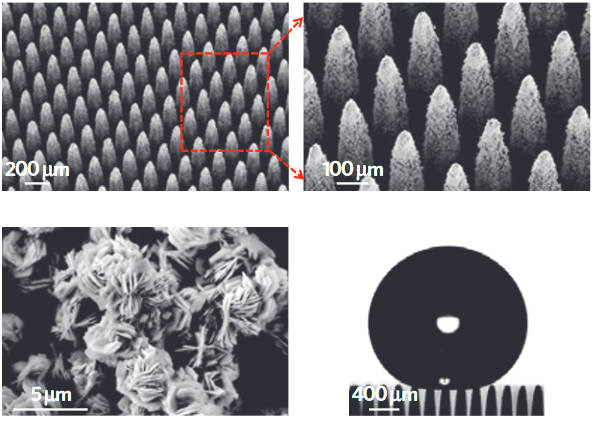}
\includegraphics[width=0.29\linewidth]{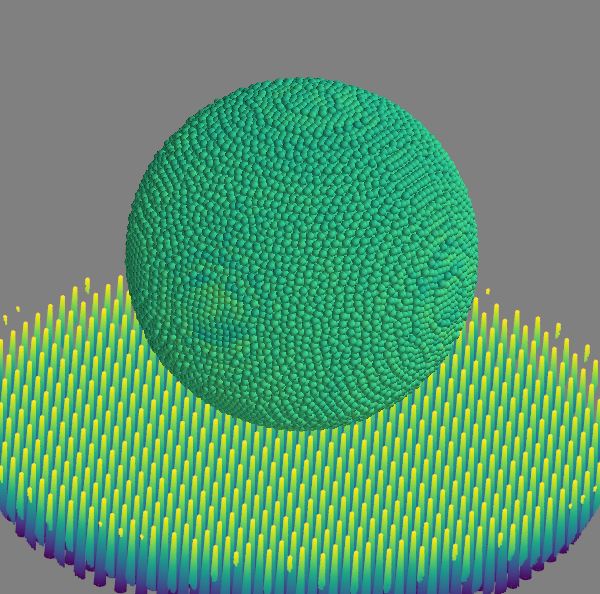}
\caption{Droplet impact on microcone arrays: experimental observation (left; Copyright 2014 Springer Nature) and SPH simulation (right)}
\label{droplet_impact}
\end{figure}
In fact, due to  the non-smoothness of the substrates, it is hard to determinate  the local surface normals $\boldsymbol{n}_i$ and vertical distance $d_p$.
Besides, the virtual particle method is a widely used boundary treatment technique to address kernel truncation and wall penetration issues.
However, when used  in isolation,  the virtual particle method  often fails to reproduce physically accurate repulsive forces during droplet rebound.
Consequently, the development of a robust contact algorithm capable of handling intricate microcone geometries is critical for reliable simulation of pancake bouncing dynamics.

The aim of this paper is to develop a robust and universally applicable numerical methods to capture pancake bouncing on superhydrophobic microcone arrays.
In our work, the  cohesive forces is  introduce to simulate intermolecular attraction.  Through theoretical derivation, we establish a relationship between cohesive forces and interfacial tension coefficients to ensure the physical consistency of the model. To address the contact issues of complex microcone structures, we define nonlocal contact repulsion in the boundary region to achieve stable simulation of pancake bouncing behavior in three-dimensional space.
Our main contributions can be summarized as follows:
\begin{enumerate}
    \item [(i)]\ {\bf A nonlocal mathematical framework} for droplet dynamics on micro-structured surfaces, rigorously linking intermolecular forces to the macroscopic surface tension coefficient without empirical parameter tuning.
    \item [(ii)]\ {\bf A stabilized nonlocal contact algorithm} for rough substrates, addressing the geometric challenges of micro-cone arrays by introducing a robust fluid-solid repulsion force model that ensures numerical stability at high Weber numbers.
    \item [(iii)]\ {\bf High-fidelity SPH-based simulations} of pancake bouncing on superhydrophobic micro-cone arrays, validated against experimental data, which elucidate the role of microstructure interactions in droplet dynamics.
\end{enumerate}

The rest of this paper is structured as follows. Section 2 presents the theoretical formulation of the governing equations for droplet impact dynamics.  Section 3 develops a robust SPH discretization method. Section 4 provides comprehensive validation of the numerical model through systematic comparisons with experimental data. Finally, Section 5 summarizes the principal conclusions and their significance.

\section{Governing equations}
The bouncing dynamics of the droplet are governed by conservation equations of mass and momentum. The the Lagrangian-form of these equations can be expressed as:
\begin{align}
\frac{\mathrm{d} \rho}{\mathrm{d} t}&=-\rho\nabla\cdot\boldsymbol{u},\\
\rho\frac{\mathrm{d} \boldsymbol{u}}{\mathrm{d} t}&=-\nabla p   +  \mu \nabla^2 \boldsymbol{u}+  \rho\boldsymbol{g} + \boldsymbol{f}_{\mathrm{att}} + \boldsymbol{F},
\end{align}
The governing equations involve the fluid velocity field $\boldsymbol{u}$, hydrodynamic pressure $p$, and local densities $\rho$. The pressure  $p$  is calculated by the following state of equation
\begin{equation}p=\frac{c^2\rho_0}{\gamma}\left(\left(\frac{\rho}{\rho_0}\right)^{\gamma}-1\right). \end{equation}
We adopt $\gamma=7$ \cite{monaghan1994simulating} for the equation of state, with $\rho_0$ and $c$ denoting the reference densities and the numerical speed of sound, respectively.
$\frac{\mathrm{d}  }{\mathrm{d} t}$ represents the material derivative where $\frac{\mathrm{d}  }{\mathrm{d} t} \equiv \frac{\partial  }{\partial t} + \boldsymbol{u}\cdot\nabla$.
$\mu$ denotes  dynamic viscosity.
The body forces consist of gravitation  $\rho\boldsymbol{g}$, droplet cohesive attractive force $\boldsymbol{f}_{\mathrm{att}}$, and contact repulsion $\boldsymbol{F}$ at the boundary interface.

Accurate characterization of cohesive attractive and repulsive interactions is crucial for reliably simulating droplet rebound dynamics on microcone arrays surfaces.
In this section, we present a unified mathematical framework that integrates smoothed kernel functions with nonlocal operators to describe both $\boldsymbol{f}_{\mathrm{att}}$ and $\boldsymbol{F}$, enabling robust three-dimensional modeling of bouncing behavior.

\subsection{Modeling of droplet cohesion}

It is well-established that the surface tension of water droplets originates from intermolecular forces (hydrogen bonding and van der Waals interactions) between water molecules.
While intermolecular interactions are physically straightforward, they impose prohibitive computational costs in macroscopic system simulations.
For surface tension modeling at continuum scales, cohesive force models offer a viable alternative.
For the simplicity of the mathematical representation, we need to introduce the smoothing kernel function $W(\boldsymbol{x}-\boldsymbol{x}',h)$,  which satisfies the following conditions:
\begin{itemize}
    \item Normalization condition:$\int_\Omega W(\boldsymbol{x}-\boldsymbol{x}',h)\textrm{d}\boldsymbol{x}'=1$;
    \item Symmetric property: $W(\boldsymbol{x}-\boldsymbol{x}',h)=W(\boldsymbol{x}'-\boldsymbol{x},h)$;
    \item Compact condition: $W(\boldsymbol{x}-\boldsymbol{x}',h)=0 \text{ where } |\boldsymbol{x}-\boldsymbol{x}'|\geq h$.
\end{itemize}
Here, $|\cdot|$ represents the Euclidean norm and $h$ is the influence radius of kernel.  $W(\boldsymbol{x},h)$ affects only a specific support domain. 
 Additionally, the following lemma on quadruple integrals is essential for our subsequent analysis.
\begin{lemma}\label{quadruple_integral}
Assume $R_0>0$. Denote 
$$G[f]:=\int_0^{R_0}  \int_{0}^{\pi} \int_{0}^{arc\cos(s/{R_0})} \int_{s/\cos(\theta)}^{R_0} f(r)\cos(\theta) r^2\sin(\theta) drd\theta d\phi d s.$$
There holds
  \begin{align}
  G[f]=\frac{ \pi}{3}\int_0^{R_0} f(r)r^3 dr.
\end{align} 
Furthermore, if $f(\|\mathbf{x}\|)=\partial_{\|\mathbf{x}\|}W(\mathbf{x},R_0)$ and $W$ is any kernel function,
there holds
$$G[f]=\frac{1}{4}.$$
\end{lemma}
The proof of Lemma \ref{quadruple_integral} can be established using techniques of multiple integration. For completeness, we provide the detailed proof in the Appendix \ref{appendixA}.

We now proceed to construct a cohesive forces model capable of generating physically accurate interfacial tension.
In our work, each Lagrangian particle represents a finite fluid volume element comprising numerous water molecules.
Between adjacent particles, we introduce an artificial cohesive force term $\boldsymbol{f}(\boldsymbol{x},\boldsymbol{y})$ to approximate the net effect of microscopic intermolecular forces. Specifically, we denote the attraction of the particle at $\boldsymbol{y}$ to the particle at $\boldsymbol{x}$ as follows
\begin{align}\label{f_c}
  \boldsymbol{f}(\boldsymbol{x},\boldsymbol{y})
  =A_\sigma \rho(\boldsymbol{y})\nabla_{\boldsymbol{x}}W(\boldsymbol{x}-\boldsymbol{y},{R_{\mathrm{att}}}),
\end{align}
where $A_\sigma$ is a parameter dependent on interfacial tension $\sigma$ and $R_{\mathrm{att}}$ is the cohesive force interaction radius. The specific form of $A_\sigma$ will be given later.
$\nabla_{\boldsymbol{x}}W(\boldsymbol{x}-\boldsymbol{y},{R_{\mathrm{att}}})$ in the same direction as $\boldsymbol{y}-\boldsymbol{x}$.
Thus, by integrating over $\boldsymbol{y}$, the total attractive force acting on the particle located at position $\boldsymbol{x}$ is
\begin{align}\label{def_fatt}
  \boldsymbol{f}_{\mathrm{att}}(\boldsymbol{x})= \int_\Omega \boldsymbol{f}(\boldsymbol{x},\boldsymbol{y})d\boldsymbol{y}
  =  A_\sigma \int_\Omega\rho(\boldsymbol{y})\nabla_{\boldsymbol{x}}W(\boldsymbol{x}-\boldsymbol{y},{R_{\mathrm{att}}})d\boldsymbol{y},
\end{align}
which implies  that the attractive force scales with the nonlocal gradient of density.

\begin{figure}[!b]
  \centering
\begin{tikzpicture}[scale=1.5]
  \draw[dashed] (0,0,-0.5) -- (1,0,-0.5);
  \draw[dashed] (1,0,-0.5) -- (1,1,-0.5);
  \draw (1,1,-0.5) -- (0,1,-0.5);
  \draw[dashed] (0,1,-0.5) -- (0,0,-0.5); %
  \draw (0,0,2.5) -- (1,0,2.5) -- (1,1,2.5) -- (0,1,2.5) -- cycle; %
  \draw[dashed] (0,0,-0.5) -- (0,0,0 ); %
  \draw[dashed] (1,0,-0.5) -- (1,0,-0.0);
  \draw(1,0,-0.0) -- (1,0,2.5);
  \draw (1,1,-0.5) -- (1,1,2.5);
  \draw (0,1,-0.5) -- (0,1,2.5);

  \fill[blue, opacity=0.2] (0,0,2.5) -- (1,0,2.5) -- (1,1,2.5) -- (0,1,2.5) -- cycle;
  \fill[blue, opacity=0.2] (0,1,-0.5) -- (1,1,-0.5) -- (1,1,2.5) -- (0,1,2.5) -- cycle;
  \fill[blue, opacity=0.2] (1,0,-0.5) -- (1,0,2.5) -- (1,1,2.5) -- (1,1,-0.5) -- cycle;

  \fill[red, opacity=0.5] (2.1,1,-0.5) -- (1,1,-0.5) -- (1,1,2.5) -- (2.1,1,2.5) -- cycle;
  \node at (1.7,1.5,0.5) [red]{$V_1$};
  \node at (0,0,1.7) [blue]{$V_2$};

  \filldraw (1,1,1) circle (1pt) node[above right]{$\boldsymbol{x}_0$};
  \draw[->,red,thick](1,1,1)--(0.5,1,1) node[above]{$\boldsymbol{t}$};


  \begin{scope}[shift={(1,1,1)}, canvas is xy plane at z=0]
    \draw (0:1) arc (0:180:1);
    \draw (0:1) arc (0:-90:1);
    \draw[ dashed] (180:1) arc (180:270:1);
  \end{scope}

    \begin{scope}[canvas is xz plane at y=1]
    \draw (1,1) circle (1);
  \end{scope}

\draw[|<->|, dashed] (0,-0.2,2.5) -- node[below] {$R_{\mathrm{att}}$} (1,-0.2,2.5);
\draw[|<->|, dashed] (1,-0.2,2.5) -- node[below] {$R_{\mathrm{att}}$} (2,-0.2,2.5);
\draw[|<->|, dashed] (-0.2,1,-0.5) -- node[above] {$L$} (-0.2,1,2.5);
\end{tikzpicture}
\caption{Schematic of cohesion at the gas-liquid contact surface}
\label{Schematic_cohesion}
\end{figure}
The precise mathematical formulation of $A_\sigma$ is of critical importance, as its rigorous definition serves as the foundation for ensuring both the algorithm's generalizability across diverse scenarios and its reliability in practical applications.
Now we will establish a quantitative relationship between the parameter $A_\sigma$ and surface tension $\sigma$ through theoretical derivation.
Notably, the droplet's curvature radius $R$ is typically much larger than the cohesive force interaction radius $R_{\mathrm{att}}$. This scale separation justifies our model simplification that treats the liquid-gas interface as locally planar within the $R_{\mathrm{att}}$ domain (i.e., $R \rightarrow \infty$ approximation). While this planar approximation introduces minor errors in the interaction integrals, the accuracy can be systematically improved by reducing $R_{\mathrm{att}}$ through mesh refinement.
Therefore, we can determine the relationship between $A_\sigma$ and surface tension $\sigma$ by studying the tension on a flat surface.
Let $\varepsilon$ denote the thickness of the surface. At any given point $\boldsymbol{x}_0$ on the surface, consider an arbitrary unit tangent vector $\boldsymbol{t}$ at $\boldsymbol{x}_0$.
We use $V_1$ to describe the surface particles on one side of the normal plane of $\boldsymbol{t}$, and $V_2$ to denote all particles on the opposite side. Here,
$V_1$ has a thickness of $\varepsilon$, while both  $V_1$ and  $V_2$  are strip-shaped regions of length $L(L\gg {R_{\mathrm{att}}})$ (see Fig \ref{Schematic_cohesion}). Indeed, the precise relationship between $A_\sigma$ and $\sigma$ can be formally established through the following theorem.

\begin{theorem}[{\bf Effective surface tension generation}]\label{Th1}
Consider a fluid domain $\Omega$ with constant density $\rho(\boldsymbol{x}) = \rho_0$ for all $\boldsymbol{x} \in \Omega$. Let $\overline{f}_{\mathrm{att}}$ denote the magnitude of the resultant attractive force along direction $\boldsymbol{t}$ exerted by the fluid in region $V_2$ on the surface of $V_1$, defined as
\begin{align}\label{magofF}
\overline{f}_{\mathrm{att}} = \int_{V_1} \int_{V_2} \left[ \boldsymbol{f}(\boldsymbol{x}, \boldsymbol{y}) \cdot \boldsymbol{t} \right]_+ \, d\boldsymbol{y} \, d\boldsymbol{x},
\end{align}
where $[\cdot]_+$ denotes the positive part, and $\boldsymbol{f}$ is the pairwise force given in \eqref{f_c}.
If the coefficient $A_\sigma$ in the \eqref{f_c} takes the form
\begin{align}\label{new_f}
A_\sigma = \frac{4\sigma}{\rho_0 \varepsilon} \quad \text{such that}\quad  \boldsymbol{f}(\boldsymbol{x},\boldsymbol{y})
  =\frac{4\sigma}{\rho_0\varepsilon} \rho(\boldsymbol{y}) \nabla_{\boldsymbol{x}}W(\boldsymbol{x}-\boldsymbol{y},{R_{\mathrm{att}}}),
\end{align}
then the attractive force satisfies
\begin{align}\label{result}
\overline{f}_{\mathrm{att}} = \sigma L.
\end{align}
This establishes that the cohesive force $\boldsymbol{f}$ generates an effective fluid interface with surface tension coefficient $\sigma$.
\end{theorem}

\begin{proof}
Starting from the definition \eqref{magofF} and the artificial cohesive force \eqref{f_c}, we obtain
\begin{align}\label{eqsimgaL}
\overline{f}_{\mathrm{att}}
&= A_\sigma \rho_0 \int_{V_1} \int_{V_2} \left[ \nabla_{\boldsymbol{x}} W(\boldsymbol{x} - \boldsymbol{y}, R_{\mathrm{att}}) \cdot \boldsymbol{t} \right]_+ \, d\boldsymbol{y} \, d\boldsymbol{x} \notag \\
&= A_\sigma \rho_0 \int_{V_1} \int_{V_2} \left[ f(|\boldsymbol{x} - \boldsymbol{y}|) \frac{\boldsymbol{x} - \boldsymbol{y}}{|\boldsymbol{x} - \boldsymbol{y}|} \cdot \boldsymbol{t} \right]_+ \, d\boldsymbol{y} \, d\boldsymbol{x},
\end{align}
where we used the relation $f(|\boldsymbol{x}|) = \partial_{|\boldsymbol{x}|} W(\boldsymbol{x}, R_{\mathrm{att}})$ and $\nabla_{\boldsymbol{x}} |\boldsymbol{x}| = \boldsymbol{x}/|\boldsymbol{x}|$.

Let $\theta$ be the angle between $\boldsymbol{t}$ and $\boldsymbol{y} - \boldsymbol{x}$. The integrand simplifies to
\[
\frac{\boldsymbol{x} - \boldsymbol{y}}{|\boldsymbol{x} - \boldsymbol{y}|} \cdot \boldsymbol{t} = -\cos \theta.
\]
By exploiting the symmetry of $V_1$, the volume integral of $\boldsymbol{x}$ in \eqref{eqsimgaL} can be reduced to a line integral via $\boldsymbol{x} = \boldsymbol{x}_0 - s \boldsymbol{t}$.
Furthermore, adopting spherical coordinates $(r, \theta, \phi)$ centered at $\boldsymbol{x}$ with the transformation
\[
\boldsymbol{y} = \boldsymbol{x} + r
\begin{pmatrix}
-\cos \theta \\
\sin \theta \cos \phi \\
-\sin \theta \sin \phi
\end{pmatrix} \text{ and } r = |\boldsymbol{x} - \boldsymbol{y}|,
\]
\eqref{eqsimgaL} can be reformulated as
\begin{align}\label{eqsimgaL2}
\overline{f}_{\mathrm{att}}
&= -A_\sigma \rho_0 L \varepsilon \int_{0}^{R_{\mathrm{att}}} \int_{0}^{\pi} \int_{0}^{\arccos(s/R_{\mathrm{att}})} \int_{s/\cos \theta}^{R_{\mathrm{att}}} f(r) \cos \theta  r^2 \sin \theta  dr d\theta  d\phi  ds \nonumber\\
&:=-A_\sigma \rho_0 L \varepsilon G_f,
\end{align} 
where $G_f$ represents the quadruple integral expression. 
By use of Lemma \ref{quadruple_integral}, we obtain the simplified form
\begin{align}\label{q_integral}
G_f=\frac{ \pi}{3}\int_0^{R_{\mathrm{att}}} f(r)r^3 dr.
\end{align}
Noting that $f(|\boldsymbol{x}|)=\partial_{|\boldsymbol{x}|}W(\boldsymbol{x},R_{\mathrm{att}})$, it follows from Lemma \ref{quadruple_integral} that 
\begin{align}\label{result_of_Gf}
G_f=-\frac{1}{4}.
\end{align} 
Substituting \eqref{result_of_Gf} into \eqref{eqsimgaL2} gives the intermediate result
\begin{align}\label{eqnew}
\overline{f}_{\mathrm{att}} = \frac{\varepsilon L A_\sigma \rho_0}{4}.
\end{align}
Finally, by incorporating the expression for $A_\sigma$ from \eqref{new_f} into \eqref{eqnew}, we arrive at the desired result \eqref{result}. 

The physical interpretation of this result follows directly from the definition of surface tension: the interfacial force per unit length between regions $V_1$ and $V_2$ is exactly $\overline{f}_{\mathrm{att}}/L = \sigma$. This demonstrates conclusively that the pairwise cohesive force $\boldsymbol{f}(\boldsymbol{x}, \boldsymbol{y})$ defined in \eqref{new_f} generates an effective surface tension with coefficient $\sigma$ at the interface between the fluid regions. This completes the proof.
\end{proof}

%

\begin{remark}
Theorem \ref{Th1} is established under the assumption of zero interface curvature. 
Nevertheless, this approach maintains reasonable accuracy with a mathematically verifiable relative error below 4.5\% for scenarios where  the radius of curvature is no less than $R_{\mathrm{att}}$. The relative error bound was confirmed through  the selection of cubic spline kernel functions and numerical integration of Eq. \eqref{eqsimgaL}. 
To avoid unnecessary computational complexity associated with curvature corrections, we retain the original cohesive force formulation \eqref{new_f} without modification.
\end{remark}

\subsection{Modeling of substrate contact repulsion}
The fluid-solid interaction poses significant computational challenges in particle-based simulations. While conventional no-penetration conditions suffice for macroscopic phenomena like dam-break flows, they become inadequate for microscopic interfacial processes such as droplet bouncing. This necessitates precise characterization of repulsive forces on fluid particles, particularly when dealing with geometrically complex boundaries.

In conventional contact algorithms, computing repulsive forces (using Equation \eqref{repulsive}) involves both the substrate's normal direction and the vertical distance. However, this requirement becomes overly restrictive for substrates with complex geometries, such as the one illustrated in Figure \ref{droplet_impact}. {A major limitation of this method is its inefficiency in handling surfaces with high curvature. Computational performance is constrained by a critical trade-off: achieving accuracy requires fine spatial discretization, which entails significant computational cost, whereas coarse discretization leads to instability due to inaccurate definitions of the surface normal and distance, a challenge particularly acute in concave regions.}
To overcome these limitations, we propose an alternative approach that computes repulsive forces directly between particle pairs. In fact, this strategy effectively preserves the characteristics of the microcone array.  Through the implementation of parallel computing techniques, we maintain manageable computational costs while achieving significantly improved stability and broader applicability across various geometric configurations.

{
With this in mind, we now introduce the definition of the repulsive force for our model. 
Consider two particles located at positions $\boldsymbol{x}$ and $\boldsymbol{y}$, with particle $\boldsymbol{x}$ being subjected to the repulsive force field generated by particle $\boldsymbol{y}$. The potential energy of particle $\boldsymbol{x}$ can be approximated by the function
$$
\Phi(r) = k_0 W(r,R_{\mathrm{rep}}),
$$
where $W(r,R_{\mathrm{rep}})$ is the Wendland quintic kernel for three-dimension,  $r = |\boldsymbol{x} - \boldsymbol{y}|$ denotes the inter-particle distance and  $R_{\mathrm{rep}}$ represents the radius of repulsive interaction. Once the kernel function $W$ is specified, $k_0>0$ is a parameter that depends only on the physical properties of the particles at $\boldsymbol{x}$ and $\boldsymbol{y}$. Dimensional analysis shows that the parameter $k_0$ has dimensions of $[M][L]^{-1}[T]^{-2}$, which corresponds to pressure or stress.
The repulsive force exerted on particle $\boldsymbol{x}$ by particle $\boldsymbol{y}$ is then given by the negative gradient of the potential:
$$
\boldsymbol{F}(\boldsymbol{x}, \boldsymbol{y}) = -k_0  \nabla_{\boldsymbol{x}} W(|\boldsymbol{x} - \boldsymbol{y}|, R_{\mathrm{rep}}).
$$ 
Let $\Gamma$ denote the set of all microcone array points, as illustrated in Figure~\ref{Schematic}.   
The total repulsive force acting on fluid particle $\boldsymbol{x}$ due to all surrounding boundary particles within a domain $\Gamma$ is obtained by integrating the pairwise force over the domain:
$$
\boldsymbol{F}(\boldsymbol{x}) = \int_\Gamma \boldsymbol{F}(\boldsymbol{x}, \boldsymbol{y})   d\boldsymbol{y}=-k_0\int_\Gamma \nabla_xW(\boldsymbol{x}-\boldsymbol{y},R_{\mathrm{rep}}) d\boldsymbol{y}.
$$
}

{While the cohesion parameter \(A_\sigma\) is directly derived from the macroscopic surface tension \(\sigma\), the repulsion parameter \(k_0\) lacks a direct macroscopic analogue.  
The parameter \(k_0\)  is a well-defined macroscopic representation of microscopic repulsive forces. In the numerical simulations presented herein, its value is robustly calibrated through controlled experiments. It serves the critical numerical function of preventing unphysical boundary penetration.
}

\begin{figure}[!t]
  \centering
 \begin{tikzpicture}
\draw (0,1.5) circle (1);

    \draw (-0.5,0) arc[start angle=0, end angle=180, radius=0.2];
    \draw (-0.5,0) -- (-0.4,-0.9);
    \draw (-0.9,0) -- (-1,-1);

    \draw (0.2,0) arc[start angle=0, end angle=180, radius=0.2];
    \draw (0.2,0) -- (0.3,-0.9);
    \draw (-0.2,0) -- (-0.3,-0.9);

    \draw (0.9,0) arc[start angle=0, end angle=180, radius=0.2];
    \draw (0.9,0) -- (1,-1);
    \draw (0.5,0) -- (0.4,-0.9);

    \draw  (-1,-1) -- (1,-1); 
    \draw  (-0.4,-0.9)  -- (-0.3,-0.9);
    \draw  (0.3,-0.9) -- (0.4,-0.9);

    \draw[->, dashed] (0,1.5) -- node[above left] {$R$} (0.9,1.9); 
\node at (0,-0.8)   {$\Gamma$};
\node at (0.5,1)   {$\Omega$};
\end{tikzpicture}
\caption{Schematic diagram of a droplet with a substrate}
\label{Schematic}
\end{figure}

\subsection{Governing equations with nonlocal interface and boundary interactions}

With the interfacial tension--cohesion relationship quantified and repulsive forces expressed in a nonlocal form, we now integrate these formulations with the conservation laws of mass and momentum to derive the governing equations of the system:
\begin{align}
    \frac{\mathrm{d} \rho}{\mathrm{d} t} &= -\rho \nabla \cdot \boldsymbol{u}, \\
    \frac{\mathrm{d} \boldsymbol{u}}{\mathrm{d} t} &= -\frac{\nabla p}{\rho} + \frac{\mu}{\rho} \nabla^2 \boldsymbol{u} + \boldsymbol{g}
    + \frac{4\sigma}{\rho^2 \varepsilon} \int_\Omega \rho(\boldsymbol{y}) \nabla_{\boldsymbol{x}} W(\boldsymbol{x} - \boldsymbol{y}, R_{\mathrm{att}}) \, \mathrm{d}\boldsymbol{y} \\
    &\quad - \frac{1}{\rho} k_0\int_\Gamma \nabla_{\boldsymbol{x}} W(\boldsymbol{x} - \boldsymbol{y}, R_{\mathrm{rep}}) \, \mathrm{d}\boldsymbol{y}, \notag\\
    \frac{\mathrm{d} \boldsymbol{x}}{\mathrm{d} t} &= \boldsymbol{u},
\end{align}
where the pressure $p$ is determined by the equation of state
\begin{equation}
    p = \frac{c^2 \rho_0}{\gamma} \left( \left( \frac{\rho}{\rho_0} \right)^\gamma - 1 \right).
\end{equation}

\section{Numerical Methods}
\subsection{Foundations of SPH Methodology}
In the Smoothed Particle Hydrodynamics framework, the \textit{kernel approximation} (also termed \textit{integral approximation}) of a scalar field $f(\boldsymbol{x})$ is given by
\begin{equation}\label{eq2}
    f_I(\boldsymbol{x}) := \int_\Omega f(\boldsymbol{x}') W(\boldsymbol{x} - \boldsymbol{x}', h) \, \mathrm{d}\boldsymbol{x}' \approx f(\boldsymbol{x}),
\end{equation}
where $W$ denotes the smoothing kernel function with support radius $h$. Introducing the normalized distance $q = |\boldsymbol{x} - \boldsymbol{y}|/(h/2)$, we consider two widely used kernel functions and one modified kernel function:

1). \textbf{Cubic Spline Kernel} \cite{monaghan1992smoothed}:
\begin{align}
    W^{\text{cs}}(\boldsymbol{x}-\boldsymbol{y},h) = \tilde{W}^{\text{cs}}(q) =
    \begin{cases}
        \sigma_1 \left[1 - \frac{3}{2}q^2 \left(1 - \frac{q}{2}\right)\right], & 0 \leq q \leq 1, \\
        \sigma_1(2 - q)^3/{4}, & 1 < q \leq 2, \\
        0, & q > 2,
    \end{cases}
\end{align}
where $\sigma_1 = [\pi (h/2)^3]^{-1}$ for three-dimension.

2). \textbf{Wendland Quintic Kernel} \cite{gomez2010state}:
\begin{align}
    W^{\text{wq}}(\boldsymbol{x}-\boldsymbol{y},h) = \tilde{W}^{\text{wq}}(q) =
    \begin{cases}
        \sigma_2 (1 - q/2)^4 (2q + 1), & 0 \leq q \leq 2, \\
        0, & q > 2,
    \end{cases}
\end{align}
with $\sigma_2 = 21/[16\pi (h/2)^3]$ for three-dimensions. The Wendland Quintic Kernel function can constitute a good choice in terms of computational accuracy and effectiveness, since it provides a higher order of interpolation with a computational cost comparable to the quadratic kernel.

{
3). \textbf{New Wendland Kernel}:
\begin{align}
    W^{\text{new}}(\boldsymbol{x}-\boldsymbol{y},h) = \tilde{W}^{\text{new}}(q) =
    \begin{cases}
        \sigma_3 (1 - q/2)^5 (5q/2 + 1), & 0 \leq q \leq 2, \\
        0, & q > 2,
    \end{cases}
\end{align}
with $\sigma_3 = 189/[128\pi (h/2)^3]$ for three-dimension. The newly proposed kernel function is a straightforward modification of the Wendland Quintic Kernel, achieved by increasing the polynomial degree, which brings it closer to approximating a repulsive potential. In practice, such a kernel is relatively more effective at capturing the characteristics of short-range forces. 
}

Following the idea of Eq. \eqref{eq2}, we can readily derive an approximation of the derivative function using integration by parts
\begin{equation}\label{eq3}
[D^\beta f]_I(\boldsymbol{x}) = -\int_\Omega f(\boldsymbol{x}') D^\beta_{\boldsymbol{x}'} W(\boldsymbol{x}-\boldsymbol{x}', h) \, \mathrm{d}\boldsymbol{x}',
\end{equation}
where $\beta$ denotes a multi-index characterizing the derivative order.

Discretization of the kernel approximation yields the \textit{particle approximation} for both scalar fields and vector field divergences:
\begin{align}
\langle f \rangle_i &= \sum_j \frac{m_j}{\rho_j} f_j W_{ij}^{\text{wq}}, \label{PAf} \\
\langle \nabla \cdot \boldsymbol{f} \rangle_i &= -\sum_j \frac{m_j}{\rho_j} \boldsymbol{f}_j \cdot \nabla_j W_{ij}^{\text{wq}}, \label{PADf}
\end{align}
where $f_j := f(\boldsymbol{x}_j)$, $W_{ij}^\text{wq} := W^{\text{wq}}(\boldsymbol{x}_i-\boldsymbol{x}_j,h)$ and $\nabla_jW_{ij}^\text{wq} := \nabla_{\boldsymbol{x}_j}W^{\text{wq}}(\boldsymbol{x}_i-\boldsymbol{x}_j,h)$.
Here, $i$ denotes the interpolating particle and $j$ refers to the neighbouring particles within the support. $\frac{m_j}{\rho_j}$ represents the volume of the particle $j$.

In practical applications, modified versions of Eq. \eqref{PADf} are typically employed rather than their original forms. When a derivative function is multiplied or divided by density, it can be incorporated within the summation operator. Specifically, we consider the following two identities \cite{monaghan1992smoothed}:
\begin{align}
[\rho\nabla\cdot\boldsymbol{f}](\boldsymbol{x}) &= \nabla\cdot(\rho \boldsymbol{f})(\boldsymbol{x}) - \boldsymbol{f}(\boldsymbol{x}) \cdot \nabla\rho(\boldsymbol{x}), \label{eq:div_identity} \\
\left[\frac{\nabla f}{\rho}\right](\boldsymbol{x}) &= \nabla\left(\frac{f}{\rho}\right)(\boldsymbol{x}) + \frac{f(\boldsymbol{x})}{\rho^2(\boldsymbol{x})}\nabla\rho(\boldsymbol{x}). \label{eq:grad_identity}
\end{align}
Combining Eq. \eqref{PADf}, we obtain the following commonly used SPH formulas:
\begin{align}
    \langle \rho\nabla\cdot\boldsymbol{u} \rangle_i &= -\sum_j m_j \boldsymbol{u}_{ij} \cdot \nabla_i W_{ij}^{\text{wq}},  \\
    \left\langle \frac{\nabla p}{\rho} \right\rangle_i &= \sum_j m_j \left( \frac{p_i}{\rho_i^2} + \frac{p_j}{\rho_j^2} \right) \nabla_i W_{ij}^{\text{wq} }.
\end{align}
Here, $\boldsymbol{u}_{ij}:=\boldsymbol{u}_{i}-\boldsymbol{u}_{j}$.
Another common SPH format is as follows
\begin{align}
\left\langle \frac{\mu}{\rho} \nabla^2 \boldsymbol{u}\right\rangle_i=\sum_j \frac{ m_j (\mu_i+\mu_j)}{(\rho_i+\rho_j)^2/4} \frac{\boldsymbol{x}_{ij}\cdot\nabla_iW_{ij}^\text{wq}}{|\boldsymbol{x}_{ij}|^2+(0.01 h_{ij})^2}\boldsymbol{u}_{ij}.
\end{align}

\subsection{SPH formulations}
While well-established SPH discretization schemes exist for local differential operators \cite{monaghan2005smoothed,liu2003smoothed,Smoothed000274685600002}, this work focuses on developing novel formulations for non-local operators. The SPH framework demonstrates particular advantages for non-local operators due to its intrinsic non-local nature through kernel approximations. Specifically, we have
\begin{align}
\left\langle \frac{\boldsymbol{f}_{\mathrm{att}}}{\rho} \right\rangle_i
&= \left\langle \frac{4\sigma}{\rho^2 \varepsilon} \int_\Omega \rho(\boldsymbol{y}) \nabla_{\boldsymbol{x}} W^{\text{cs}}(\boldsymbol{x}-\boldsymbol{y}, R_{\mathrm{att}}) \, \mathrm{d}\boldsymbol{y} \right\rangle_i \notag \\
&:= \frac{4\sigma}{\rho_0^2 \varepsilon} \sum_{j=1}^{N_f} m_j \nabla_i W_{ij}^{\text{cs},R_{\mathrm{att}}}, \label{eq:attractive_force}
\end{align}
and
\begin{align}
\left\langle \frac{\boldsymbol{F}}{\rho} \right\rangle_i
&= \left\langle -\frac{1}{\rho} k_0\int_\Gamma  \nabla_{\boldsymbol{x}} W^{\text{new}}(\boldsymbol{x}-\boldsymbol{y}, R_{\mathrm{rep}}) \, \mathrm{d}\boldsymbol{y} \right\rangle_i \notag \\
&:= -k_0 \sum_{k=1}^{N_s} \frac{m_k}{\rho_i \rho_k}  \nabla_i W_{ik}^{\text{new},R_{\mathrm{rep}}}, \label{eq:repulsive_force}
\end{align}
where $\nabla_iW_{ij}^{\text{cs},R_{\mathrm{att}}}:=\nabla_{\boldsymbol{x}_i}W^{\text{cs}}(\boldsymbol{x}_i-\boldsymbol{x}_j,R_{\mathrm{att}})$ and $\nabla_iW_{ij}^{\text{new},R_{\mathrm{rep}}}:=\nabla_{\boldsymbol{x}_i}W^{\text{new}}(\boldsymbol{x}_i-\boldsymbol{x}_j,R_{\mathrm{rep}})$.

Then the governing equations can be rewritten as follows. Let $N_f$ denote the number of fluid particles and $N_s$ denote the number of solid particles.
For any $1\leq i,j\leq N_f$, $1\leq k\leq N_s$, we solve the following system
\begin{align}
\frac{\mathrm{d} \rho_i}{\mathrm{d} t}=& \sum_{j=1}^{N_f+N_s}m_j  \boldsymbol{u}_{ij}\cdot\nabla_iW_{ij}^\text{wq}, \quad p_i=\frac{c^2\rho_0}{\gamma}\left(\left(\frac{\rho_i}{\rho_0}\right)^{\gamma}-1\right),\\
\frac{\mathrm{d} \boldsymbol{u}_i}{\mathrm{d} t}= & -\sum_{j=1}^{N_f}m_j\left(\frac{p_i}{\rho_i^2}+\frac{p_j}{\rho_j^2}+\Pi_{i,j}^{\text{art}}\right)\nabla_iW_{ij}^\text{wq}
 +\frac{4\sigma}{\varepsilon\rho_0^2}\sum_{j=1}^{N_f} m_j \nabla_iW_{ij}^{\text{cs},R_{\mathrm{att}}} \nonumber\\
&-\sum_{k=1}^{N_s} m_k\left(\frac{k_0}{\rho_k\rho_i}+\Pi_{i,k}^{\text{art}}\right) \nabla_iW_{ik}^{\text{new},R_{\mathrm{rep}}}
\nonumber\\
&+\sum_{j=1}^{N_f} \frac{ m_j (\mu_i+\mu_j)}{(\rho_i+\rho_j)^2/4} \frac{\boldsymbol{x}_{ij}\cdot\nabla_iW_{ij}^\text{wq}}{|\boldsymbol{x}_{ij}|^2+0.01 (h_{ij})^2}\boldsymbol{u}_{ij} +\boldsymbol{g}_i,  \\
\frac{\mathrm{d} \boldsymbol{x}_i}{\mathrm{d} t}=& (1-\epsilon)\boldsymbol{u}_{i}+\epsilon\sum_{j=1}^{N_f}m_j  \boldsymbol{u}_{ij} W_{ij}^\text{wq}.
 \end{align}
Here,  $\Pi_{a,b}^{\text{art}}$ is artificial viscosity term  proposed by Monaghan \cite{monaghan2005smoothed} and
\begin{align}
\Pi_{a,b}^{\text{art}}&=\begin{cases}
        \frac{-\alpha c\mu_{ab}}{\bar{\rho}_{ab}} &
        \boldsymbol{u}_{ab}\cdot\boldsymbol{x}_{ab}<0;\\
        0 & \boldsymbol{u}_{ab}\cdot\boldsymbol{x}_{ab}\geq0;
  \end{cases},       \mu_{ab}=\frac{h\boldsymbol{u}_{ab}\cdot\boldsymbol{x}_{ab}}{\boldsymbol{x}_{ab}^{2}+0.01h^{2}},
\bar{\rho}_{ab} = \frac{\rho_a + \rho_b}{2}.
\end{align}
We set $\epsilon=0.5$ \cite{monaghan1994simulating} and $\alpha=0.015$  to guarantee the stability of the numerical scheme.

\subsection{Time integration}
The Predict-Evaluate-Correct method is adopted for time integration, offering high computational efficiency and minimal storage demands. For each time step,  one can get $\rho_i^{n+1}$, $\boldsymbol{u}_i^{n+1}$, $\boldsymbol{x}_i^{n+1},$ $\frac{\mathrm{d} \rho_i^{n+1/2}}{\mathrm{d} t}$, $\frac{\mathrm{d} \boldsymbol{u}_i^{n+1/2}}{\mathrm{d} t}$, $\frac{\mathrm{d} \boldsymbol{x}_i^{n+1/2}}{\mathrm{d} t}$ by given $\rho_i^{n}$, $\boldsymbol{u}_i^{n}$, $\boldsymbol{x}_i^{n}$, $\frac{\mathrm{d} \rho_i^{n-1/2}}{\mathrm{d} t}$, $\frac{\mathrm{d} \boldsymbol{u}_i^{n-1/2}}{\mathrm{d} t}$, $\frac{\mathrm{d} \boldsymbol{u}_i^{n-1/2}}{\mathrm{d} t}$.
Firstly, we predict the value at time $t=t_{n+\frac{1}{2}}$ by
\begin{align}
\left\{
\begin{array}{rl}
\rho_i^{n+1/2}&=\rho_i^n+\frac{\Delta t}{2}\cdot \frac{\mathrm{d} \rho_i^{n-1/2}}{\mathrm{d} t},\\
\boldsymbol{u}_i^{n+1/2}&=\boldsymbol{u}_i^n+\frac{\Delta t}{2}\cdot \frac{\mathrm{d} \boldsymbol{u}_i^{n-1/2}}{\mathrm{d} t},\\
\boldsymbol{x}_i^{n+1/2}&=\boldsymbol{x}_i^n+\frac{\Delta t}{2}\cdot  \frac{\mathrm{d} \boldsymbol{x}_i^{n-1/2}}{\mathrm{d} t},
\end{array}\right.
\end{align}
where $\Delta t$ is the given time step.
Secondly, we evaluate the derivative value at time $t=t_{n+\frac{1}{2}}$ by
\begin{align}
\left\{
\begin{array}{rl}
\frac{\mathrm{d} \rho_i^{n+1/2}}{\mathrm{d} t}&=\frac{\mathrm{d} \rho_i}{\mathrm{d} t}(\rho_i^{n+1/2},\boldsymbol{u}_i^{n+1/2},\boldsymbol{x}_i^{n+1/2}),\\
\frac{\mathrm{d} \boldsymbol{u}_i^{n+1/2}}{\mathrm{d} t}&= \frac{\mathrm{d} \boldsymbol{u}_i}{\mathrm{d} t}(\rho_i^{n+1/2},\boldsymbol{u}_i^{n+1/2},\boldsymbol{x}_i^{n+1/2}),\\
 \frac{\mathrm{d} \boldsymbol{x}_i^{n+1/2}}{\mathrm{d} t}&=\frac{\mathrm{d} \boldsymbol{x}_i}{\mathrm{d} t}(\rho_i^{n+1/2},\boldsymbol{u}_i^{n+1/2},\boldsymbol{x}_i^{n+1/2}).
\end{array}\right.
\end{align}
Finally, we can correct the value by
\begin{align}
\left\{
\begin{array}{rl}
\rho_i^{n+1}&=\rho_i^n+\Delta t\cdot \frac{\mathrm{d} \rho_i^{n+1/2}}{\mathrm{d} t},\\
\boldsymbol{u}_i^{n+1}&=\boldsymbol{u}_i^n+\Delta t\cdot \frac{\mathrm{d} \boldsymbol{u}_i^{n+1/2}}{\mathrm{d} t},\\
\boldsymbol{x}_i^{n+1}&=\boldsymbol{x}_i^n+\Delta t\cdot \frac{\mathrm{d} \boldsymbol{x}_i^{n+1/2}}{\mathrm{d} t}.
\end{array}\right.
\end{align}
This PEC scheme ensures second-order accuracy in time while maintaining numerical stability, making it particularly suitable for long-time simulations in particle-based methods.

\section{Numerical Validation}\label{sec:numerical_examples}
This section presents systematic validation of the proposed SPH framework through comprehensive simulations of droplet impact dynamics on superhydrophobic microcone arrays. The study examines a range of Weber numbers ($\textit{We}$), defined as:
\begin{align}\label{we}
\textit{We} = \frac{\rho v_0^2 r_0}{\sigma},
\end{align}
where $\rho = 1~\mathrm{g/cm}^3$ represents the water density and $\sigma = 0.0728~\mathrm{N/m}$ denotes the surface tension coefficient. The fluid properties include dynamic viscosity $\mu = 10^{-3}~\mathrm{Pa\cdot s}$ ($10^{-6}~\mathrm{g/(mm\cdot ms)}$) and gravitational acceleration $\boldsymbol{g} = 9.81~\mathrm{m/s}^2$.  
Quantitative validation against the experimental data of \cite{liu2014pancake} (see Supplementary Movies 1–3 in their original work, available at: \url{https://www.nature.com/articles/nphys2980\#Sec7}) confirms the accuracy of our numerical framework. For direct comparison, the corresponding simulations are provided in Supplementary Movies 1–3, showing consistent spatiotemporal patterns with the experiments.

\subsection{Geometric Configuration}
The microcone array exhibits the following dimensional characteristics:
\begin{itemize}
    \item Cone height: $H = 0.8~\mathrm{mm}$
    \item Base radius: $R_{\mathrm{max}} = 0.045~\mathrm{mm}$
    \item Tip radius: $R_{\mathrm{min}} = 0.01~\mathrm{mm}$
\end{itemize}
Arranged in a chessboard pattern with center-to-center spacing $L = 0.2~\mathrm{mm}$, the microcones form a complex structured surface with high curvature gradients. The impacting water droplet has initial radius $r_0 = 1.45~\mathrm{mm}$, with impact velocity $v_0$ determined by the target Weber number via Eq.~\eqref{we}. Figure~\ref{droplet_impact} shows the initial computational domain configuration, including the droplet's initial position relative to the micro-structured surface.

\subsection{Computational Parameters}
The SPH simulation utilizes carefully selected parameters to balance numerical accuracy with computational efficiency. We set the initial particle of fluid spacing to $\Delta x_f = 0.06~\mathrm{mm}$, providing sufficient resolution to capture detailed droplet dynamics. We set the initial particle of solid  spacing to $\Delta x_s = 0.02~\mathrm{mm}$, providing sufficient resolution to capture the microstructure.  The thickness of  the surface is set as $\varepsilon=0.06$ mm. For the kernel function implementation, we employ two distinct characteristic length scales: (1) an interaction radius of $R_{\mathrm{att}} = 0.26~\mathrm{mm}$ for cohesion simulation, and (2) a smoothing length of $h = R_{\mathrm{rep}} = 0.1625~\mathrm{mm}$ for field variable calculations and repulsive force modeling. 
This configuration ensures the long-range characteristics of attractive interactions while maintaining computational efficiency. 
The numerical speed of sound is set to \( c = 10(1~\mathrm{m/s} + v_0) \), and the time step is dynamically adjusted as \( \Delta t = 0.062 R_{\mathrm{att}}/c \) to maintain numerical stability. 
The model incorporates only one adjustable parameter,  parameter $k_0$, which takes fixed values of 0.0178 $\rm g/(mm\cdot ms^2)$. The key point is that the $k_0$ value remains constant regardless of changes in the Weber number, and its single-parameter nature allows it to be easily determined through simple calibration.

\subsection{Droplet Bouncing on Horizontal Microcone Arrays}
\label{subsec:horizontal}

\subsubsection{Case 1: Low-Weber-Number Impact ($\textit{We} = 7.1$)}
The experimental setup consists of a  water droplet with a diameter $r_0$ impacting the microcone array at a velocity $v_0 = 0.594$ m/s. Snapshots of the numerical simulations are presented in Figure \ref{fig:we7_1} and are compared to the experiment. Experimental observations reveal that under low Weber number conditions, droplets exhibit conventional bouncing behavior on micro-cone arrays, characterized by contraction prior to detachment from the substrate. The numerical simulations (Supplementary Movie 1) demonstrate remarkable temporal consistency with physical experiments throughout the entire process. Figure \ref{fig:we14_2} (a) presents the temporal evolution of the droplet’s horizontal diameter and its vertical height relative to the substrate during the bouncing process.
We define $t_{\uparrow}$ as the time instant when the droplet vertical height first reaches zero, and $t_\text{contact}$ as the moment when the droplet height first attains $R_{\mathrm{rep}}$.
A distinct downward impact is observed around $t = 10$ ms, consistent with experimental observations. This impact, in fact, provides the droplet with sufficient momentum to detach from the substrate.
\begin{figure}[t!]
    \centering
    \includegraphics[width=1\linewidth]{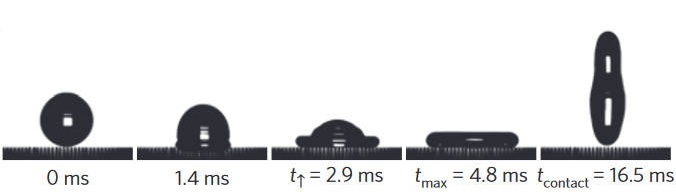} \\
    \includegraphics[width=0.19\linewidth]{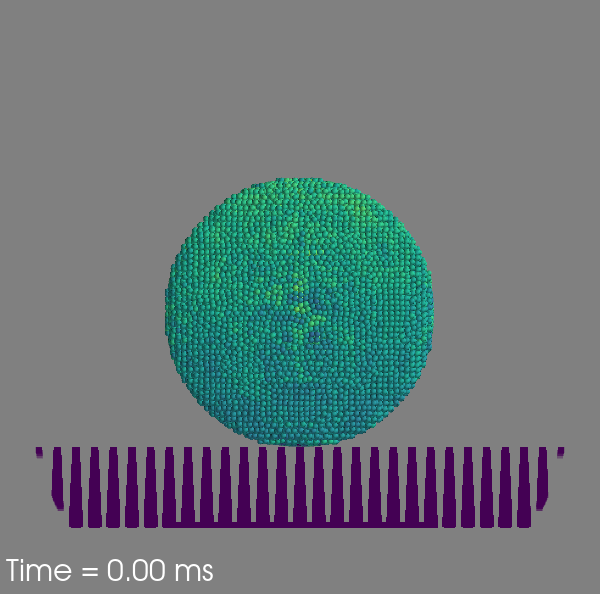}
    \includegraphics[width=0.19\linewidth]{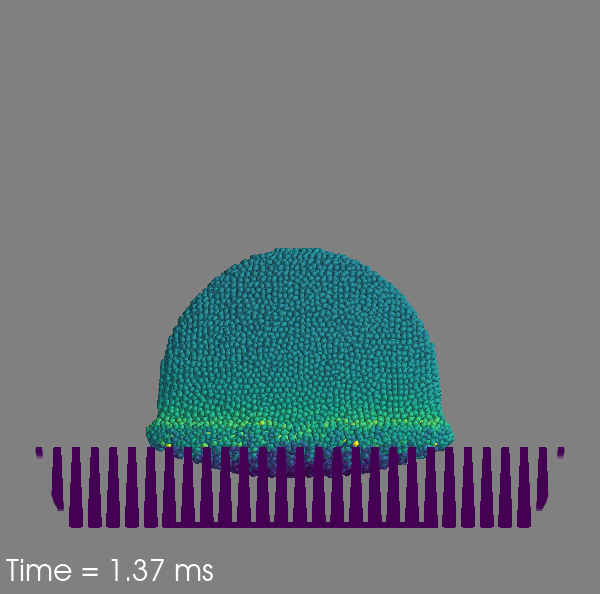}
    \includegraphics[width=0.19\linewidth]{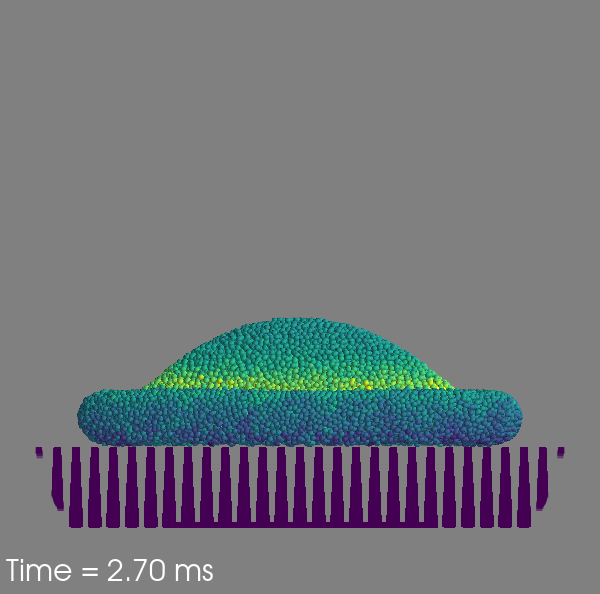}
    \includegraphics[width=0.19\linewidth]{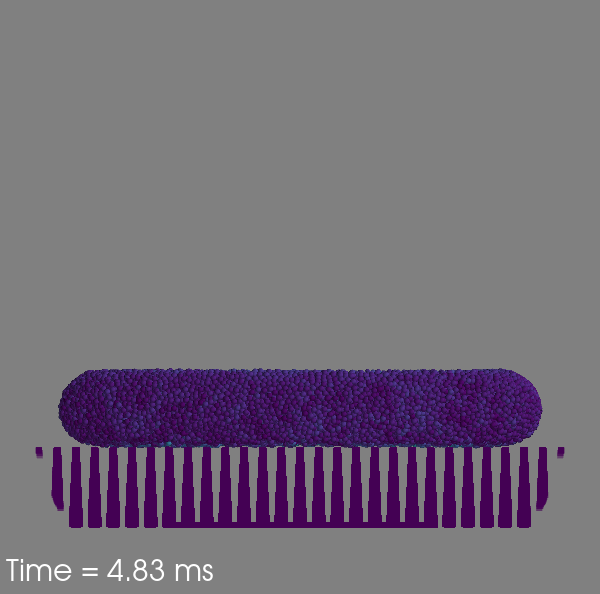}
    \includegraphics[width=0.19\linewidth]{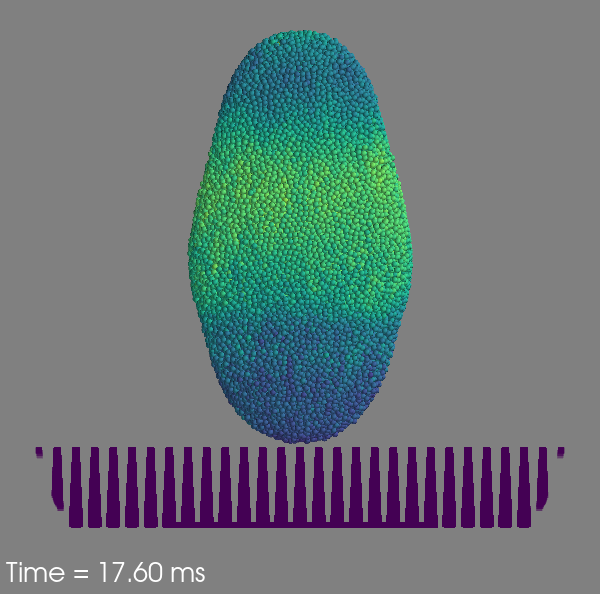}
    \caption{Comparison of simulated (bottom) and experimental (top; Copyright 2014 Springer Nature) droplet shapes during droplet bouncing at $\textit{We} = 7.1$. Time instances show   initial impact at $t = 0$ ms,   maximal spreading at $t = 4.8$ ms, and   complete rebound at $t = 16.5$ ms.}
    \label{fig:we7_1}
\end{figure}
\begin{figure}
\centering
\begin{tikzpicture}[scale=0.8]
    \begin{axis}[ylabel=Diameter and high (mm), xlabel=t(ms), legend pos=north east,
    title={(a) $We$=7.1}] 
    \addplot[blue, mark=.] table[x=Time, y=Diameter] {data/diameter_data_7.csv};
    \addlegendentry{Diameter}
    
    \addplot[red, mark=.] table [x=Time, y=High]{data/diameter_data_7.csv};
    \addlegendentry{High}
    
    \draw[blue, ->, >=stealth, thick] (axis cs:5.024,4.122) node[below] {$t_{\mathrm{max}}$} -- (axis cs:5.024,5.122);
    \draw[red, ->, >=stealth, thick] (axis cs:3, 1.00) node[above] {$t_{\uparrow}$} -- (axis cs:3.00,0.065);
    \draw[red, ->, >=stealth, thick] (axis cs:17.161,1.266) node[above] {$t_{\mathrm{contact}}$} -- (axis cs:17.161,0.266);
    \end{axis}
\end{tikzpicture} 
\begin{tikzpicture}[scale=0.8]
    \begin{axis}[ylabel=Diameter and high (mm), xlabel=t(ms), legend pos=north west,
    title={(b) $We$=14.1}]
    \addplot[blue, mark=.] table[x=Time, y=Diameter] {data/diameter_data_14.csv};
    \addlegendentry{Diameter}
    
    \addplot[red, mark=.] table [x=Time, y=High]{data/diameter_data_14.csv};
    \addlegendentry{High}
     
    \draw[blue, ->, >=stealth, thick] (axis cs:5.57, 4.98) node[below] {$t_{\mathrm{max}}$} -- (axis cs:5.57, 5.98);
    \draw[red, ->, >=stealth, thick] (axis cs:1.94, 0.057) node[left] {$t_{\uparrow}$} -- (axis cs:2.94, 0.057);
    \draw[red, ->, >=stealth, thick] (axis cs:3.82,1.177) node[above] {$t_{\mathrm{contact}}$} -- (axis cs:3.82,0.177);
    
    \end{axis}
\end{tikzpicture}
\caption{Diameter and hight of droplet with $\textit{We} = 7.1$ and $\textit{We} = 14.1$}
\label{fig:we14_2}
\end{figure}

\subsubsection{Case 2: Moderate-Weber-Number Impact ($\textit{We} = 14.1$)}
At higher impact energy ($v_0 = 0.837$ m/s), the droplet exhibits more complex dynamics. As illustrated in Figure \ref{fig:we14_1}, the droplet deviates from conventional bouncing behavior and exhibits a counterintuitive sequence: detachment from the substrate occurs prior to maximum spreading. The close agreement between numerical snapshots and experimental observations demonstrates the remarkable simulation capability of the proposed method. Furthermore, Figure \ref{fig:we14_2} (b) presents the temporal variations of droplet radius and vertical height obtained from numerical simulations (Supplementary Movie 2). A clear sequence is observed: after detaching from the substrate, the droplet continues to spread until reaching its maximum radius before subsequent contraction.
\begin{figure}[!t]
    \centering
    \includegraphics[width=\linewidth]{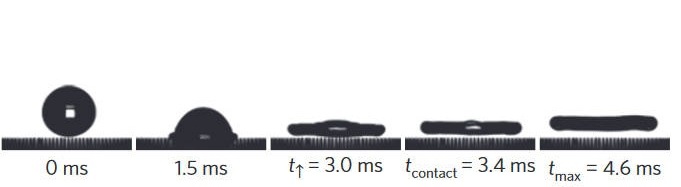} \\
    \includegraphics[width=0.19\linewidth]{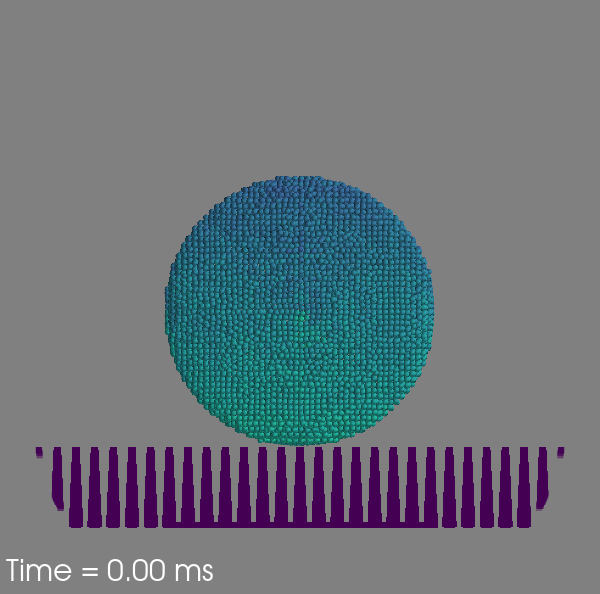}
    \includegraphics[width=0.19\linewidth]{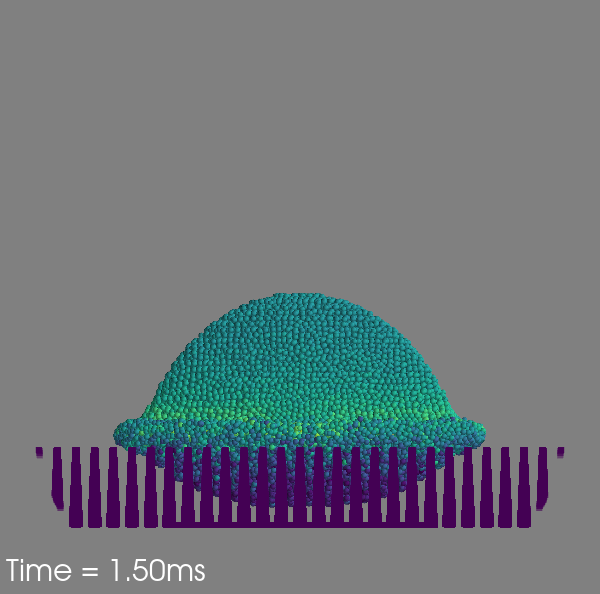}
    \includegraphics[width=0.19\linewidth]{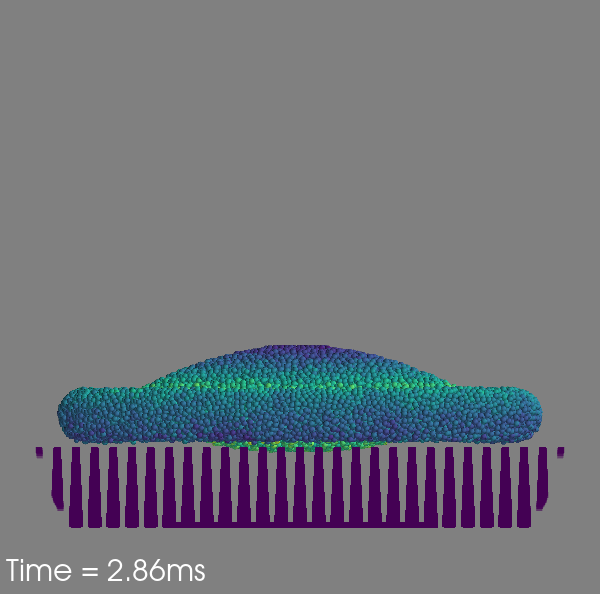}
    \includegraphics[width=0.19\linewidth]{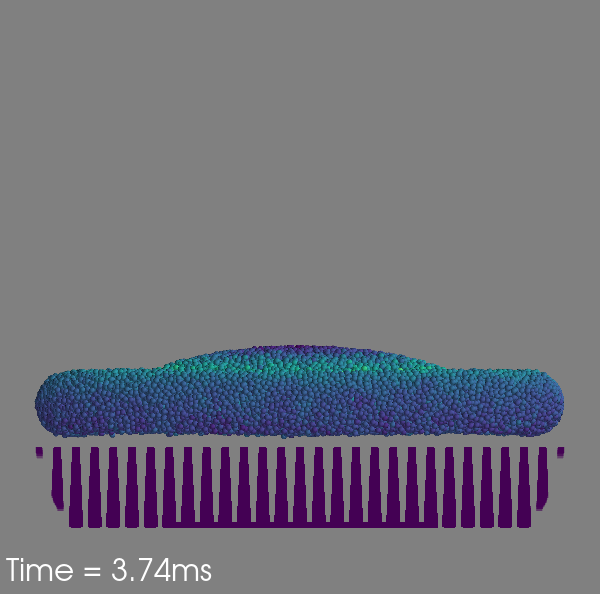}
    \includegraphics[width=0.19\linewidth]{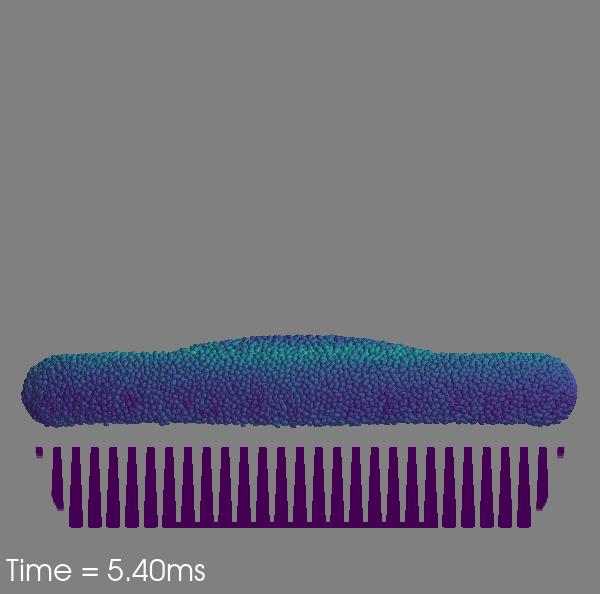}
    \caption{Comparison of simulated (bottom) and experimental (top; Copyright 2014 Springer Nature) droplet shapes during droplet bouncing at $\textit{We} = 14.1$. Time instances show initial impact at $t = 0$ ms,  complete rebound at $t = 3.4$ ms and  maximal spreading at $t = 4.6$ ms.}
    \label{fig:we14_1}
\end{figure}

\subsection{Dynamic Evolution Analysis ($\textit{We} \in [7, 24]$)}
To further validate the universality of the present method, we conducted simulations of droplet bouncing across a Weber number range of 7 to 24 and systematically compared the results with experimental data. To quantitatively assess the pancake effect, we introduce pancake quality as
\begin{equation}
Q = \frac{d_{\text{jump}}}{d_{\text{max}}},
\end{equation}
where \( d_{\text{jump}} \) is the droplet diameter at detachment time $ t_{\text{contact}}$ and \( d_{\text{max}} \) represents the global maximum diameter during spreading.  As illustrated in Figure~\ref{Fig:weber_number_comparison}, the numerical results of pancake quality   exhibit excellent agreement with experimental measurements.
Figure \ref{Fig:weber_number_comparison} also shows the temporal characteristics under different Weber numbers, including the detachment time,  maximum spreading time and take-off time along with their experimental counterparts. The comparison demonstrates satisfactory consistency between numerical predictions and physical observations.

\begin{figure}[t!]
        \centering
            \begin{tikzpicture}[scale=0.6]
            \begin{axis}[xlabel=We, ylabel=Pancake quality, legend pos=south east]
            \addplot+[
                error bars/.cd,
                y dir=both,  
                y explicit,   
            ] table[x=We, y=Q, y error=Qerr] {data/ref.csv};
            \addlegendentry{Experimental (Q)}
            \addplot[only marks, red, mark=*] table [x=We,y=Q ]{data/record.csv};
            \addlegendentry{SPH simulation (Q)}
            \end{axis}
        \end{tikzpicture}
        \begin{tikzpicture} [scale=0.6]
            \begin{axis}[xlabel=We, ylabel= Contact duration (ms)]
            \addplot+[blue, mark=triangle,
                error bars/.cd,
                y dir=both,  
                y explicit,   
            ] table[x=We, y=contact, y error=contacterr] {data/ref.csv};
            \addlegendentry{Experimental ($t_\text{contact}$)}
            \addplot[only marks, red, mark=triangle ]  table [x=We,y=t-contact]{data/record.csv};
            \addlegendentry{SPH ($t_\text{contact}$)}
            \end{axis}
        \end{tikzpicture}
\\
        \begin{tikzpicture} [scale=0.6]
            \begin{axis}[xlabel=We, ylabel=Moments of maximum spread (ms), ymin=2, ymax=7,]
            \addplot+[blue, mark=square ,
                error bars/.cd,
                y dir=both,  
                y explicit,  
            ] table[x=We, y=max, y error=maxerr] {data/ref.csv};
            \addlegendentry{Experimental ($t_\text{max}$)}
            \addplot[only marks, red, mark=square ] table [x=We,y=t-max]{data/record.csv};
            \addlegendentry{SPH ($t_\text{max}$)}
            \end{axis}
        \end{tikzpicture}
        \begin{tikzpicture} [scale=0.6]
            \begin{axis}[xlabel=We, ylabel=Take-off time (ms), ymin=2, ymax=7,]

            \addplot+[blue, mark=o,
                error bars/.cd,
                y dir=both,  
                y explicit,  
            ] table[x=We, y=up, y error=uperr] {data/ref.csv};
            \addlegendentry{Experimental ($t_{\uparrow}$)}
            \addplot[only marks,red, mark=o] table [x=We,y=t-up]{data/record.csv};
            \addlegendentry{SPH ($t_{\uparrow}$)}
            \end{axis}
        \end{tikzpicture}
        \caption{Comparison of droplet dynamics between SPH simulations and experimental measurements \cite{liu2014pancake}. Error bars represent standard deviations from experimental repetitions.}
        \label{Fig:weber_number_comparison}
\end{figure}
\subsection{Oblique bouncing on tilted substrates ($\textit{We} = 31.2$, $\theta = 30^\circ$)}
\label{subsec:oblique}
\begin{figure}[!t]
        \centering
        \includegraphics[width= \linewidth]{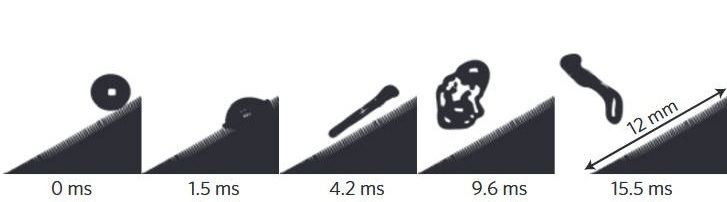} \\
        \includegraphics[width=0.19\linewidth]{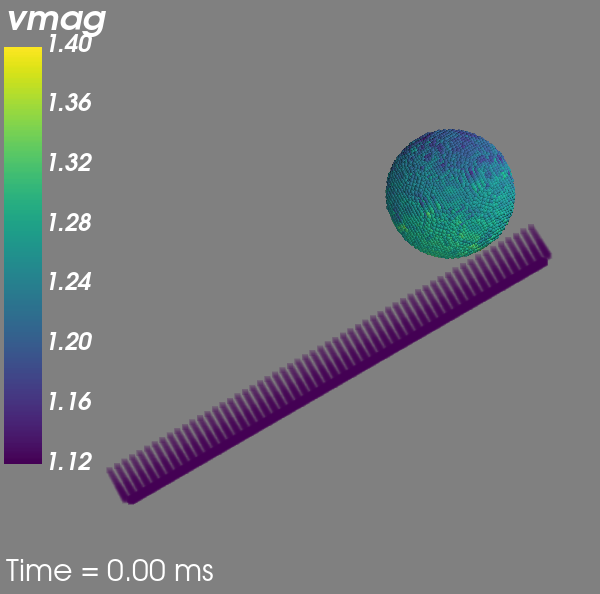}
        \includegraphics[width=0.19\linewidth]{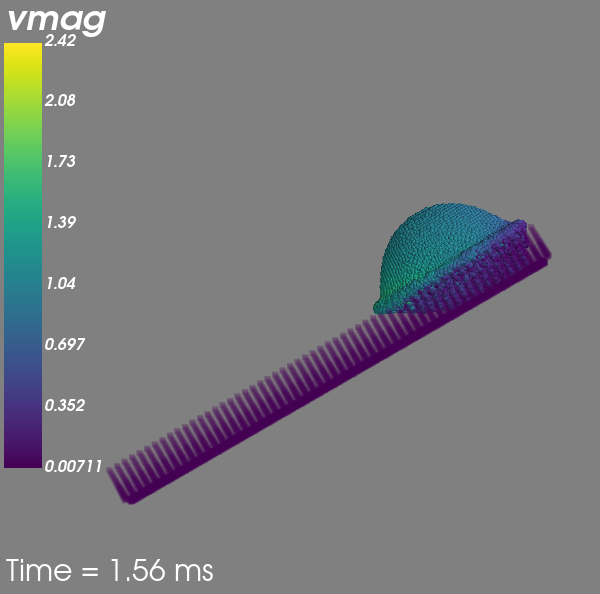}
        \includegraphics[width=0.19\linewidth]{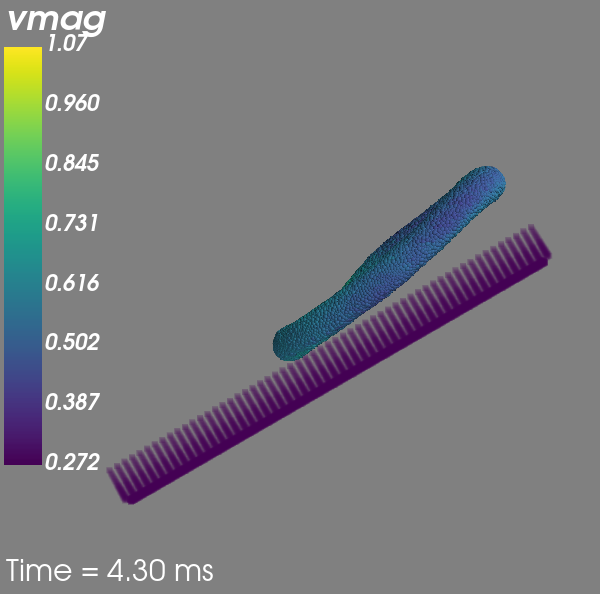}
        \includegraphics[width=0.19\linewidth]{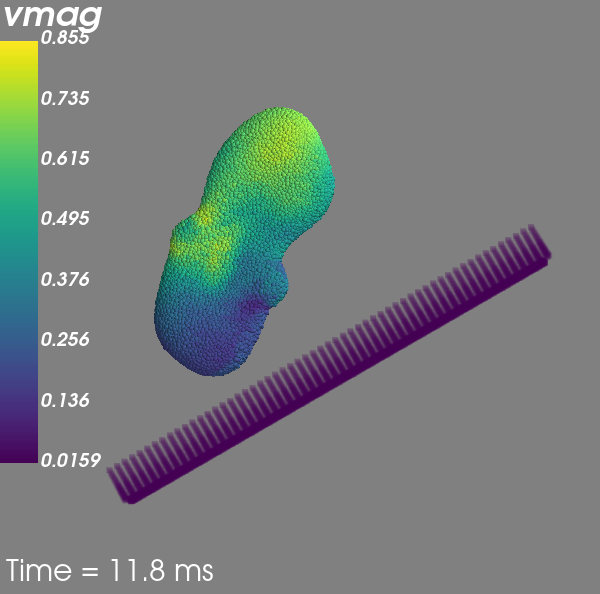}
        \includegraphics[width=0.19\linewidth]{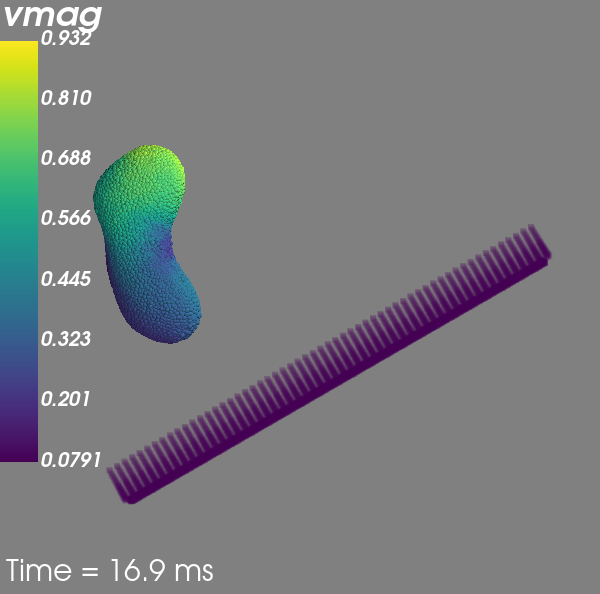}
        \caption{Comparative analysis of experimental (top; Copyright 2014 Springer Nature) and simulated (bottom) droplet morphologies at \textit{We}=31.2(velocity magnitude shown in color scale)}
    \label{fig:oblique}
\end{figure}
One practical application of superhydrophobic microcone arrays is in the development of anti-icing materials. In real-world scenarios, droplets typically impact the surface at an oblique angle rather than perpendicularly. To accurately simulate this physical phenomenon, we investigate the bouncing dynamics of water droplets on an inclined substrate. Specifically, the array plane is tilted at 30 degrees while maintaining vertical droplet impact. The simulation results are presented in Figure \ref{fig:oblique}, with experimental and numerical results shown in the top and bottom panels, respectively. Notably, the droplet detaches from the substrate in a pancake-like shape (Supplementary Movie 3), demonstrating well agreement with experimental observations.

\subsection{Secondary droplets formed by the Worthington jet ($\textit{We} = 14.1$)}
\begin{figure}[!t]
        \centering
        \includegraphics[width=0.25\linewidth,height=0.20\linewidth]{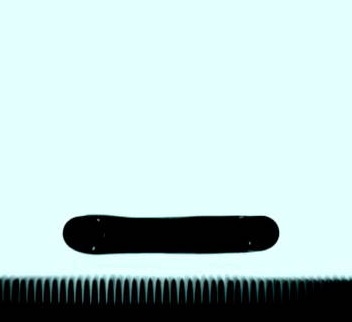}
        \includegraphics[width=0.25\linewidth,height=0.20\linewidth]{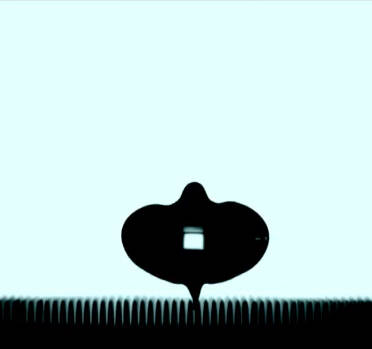}
        \includegraphics[width=0.25\linewidth,height=0.20\linewidth]{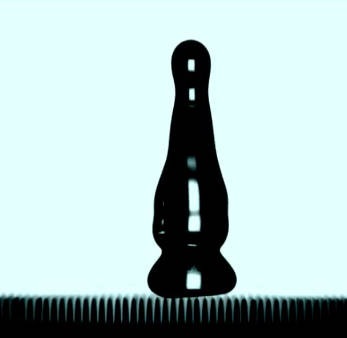}\\
        \includegraphics[width=0.25\linewidth,height=0.20\linewidth]{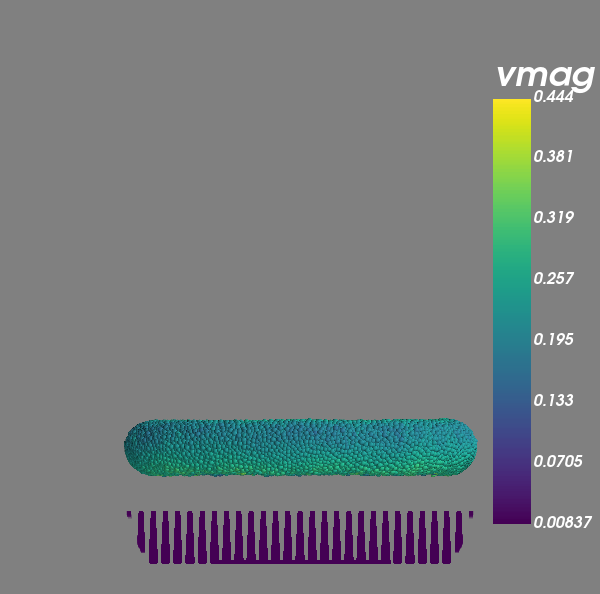}
        \includegraphics[width=0.25\linewidth,height=0.20\linewidth]{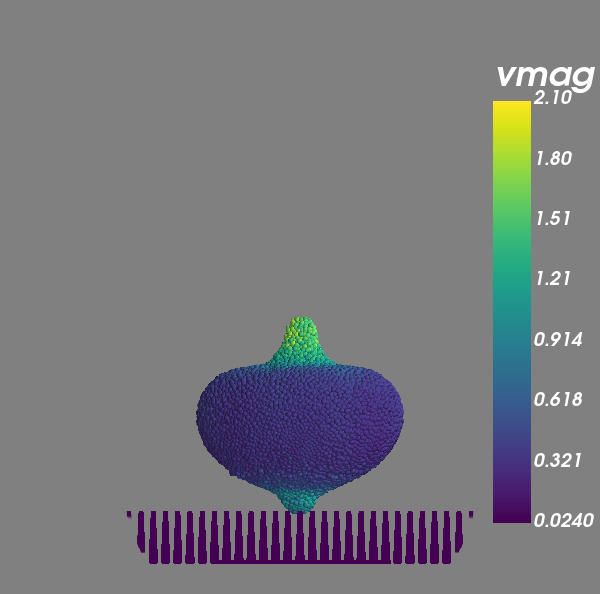}
        \includegraphics[width=0.25\linewidth,height=0.20\linewidth]{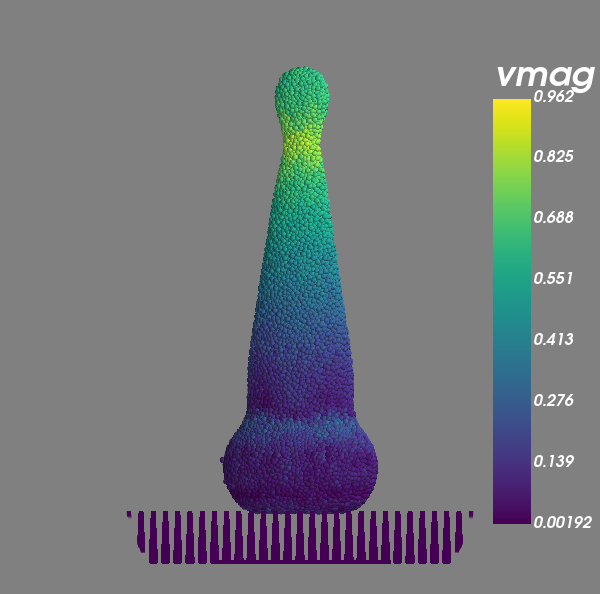} \\
        \includegraphics[width=0.25\linewidth,height=0.20\linewidth]{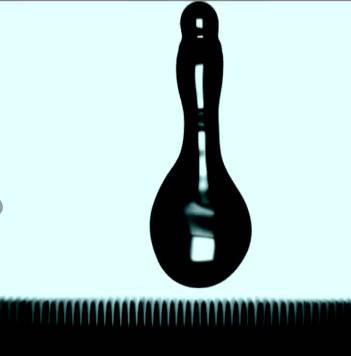}
        \includegraphics[width=0.25\linewidth,height=0.20\linewidth]{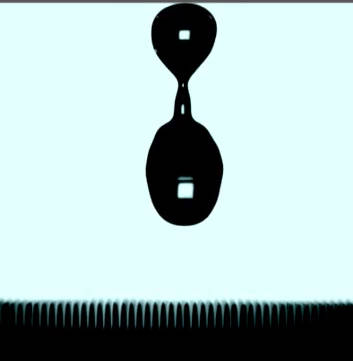}
        \includegraphics[width=0.25\linewidth,height=0.20\linewidth]{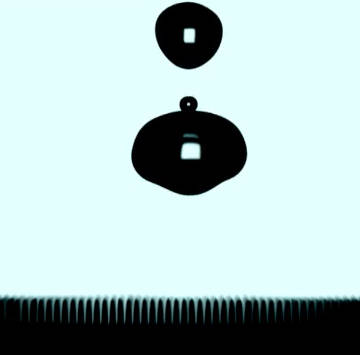}\\
        \includegraphics[width=0.25\linewidth,height=0.20\linewidth]{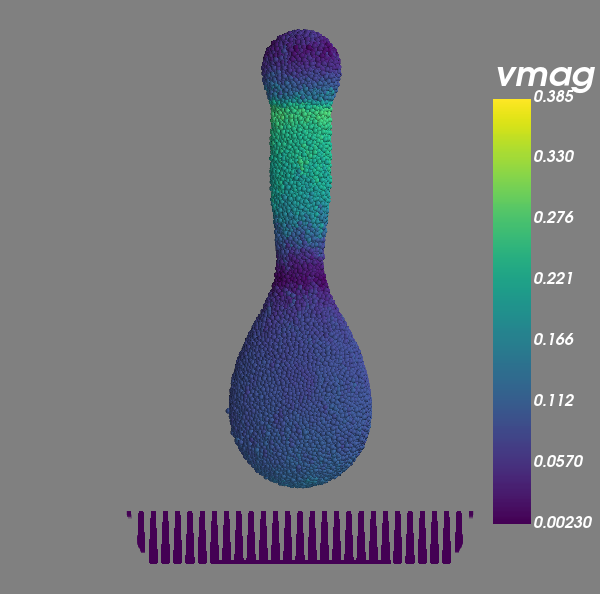}
        \includegraphics[width=0.25\linewidth,height=0.20\linewidth]{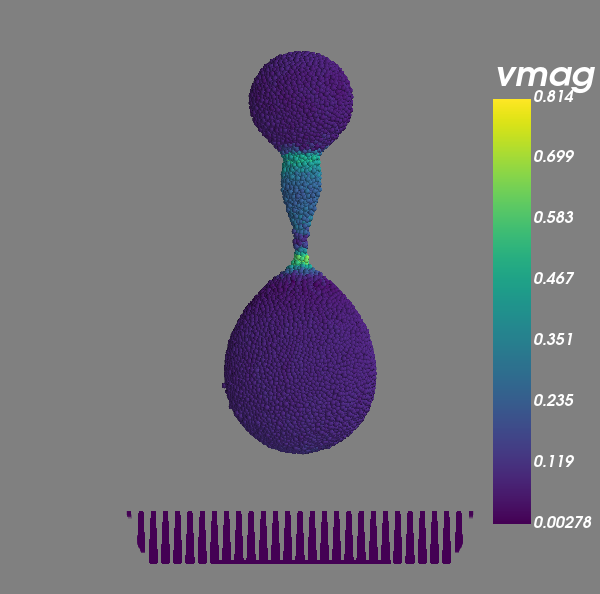}
        \includegraphics[width=0.25\linewidth,height=0.20\linewidth]{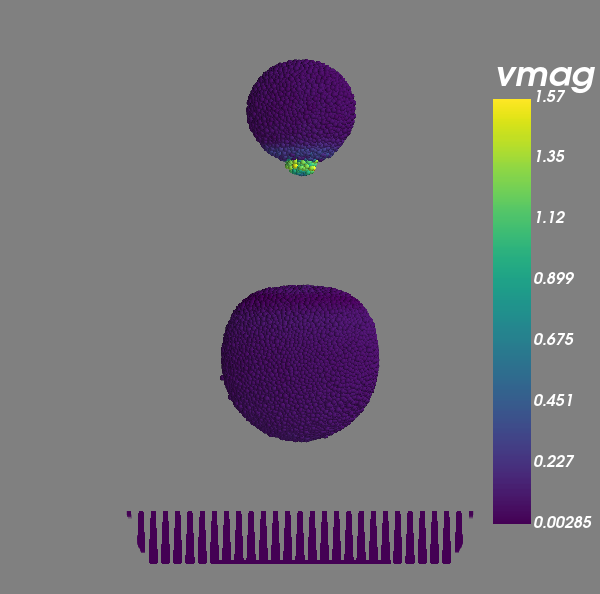}
\caption{Secondary droplet formation at \textit{We}=14.1
(top/third rows: experiment, Copyright 2014 Springer Nature; second/fourth rows: simulation, velocity magnitude in color scale)}
\label{fig:jet}
\end{figure}
Droplets with moderate Weber numbers have the potential to undergo secondary impacts subsequent to their initial interaction with the substrate. This secondary impact is often more forceful, imparting sufficient momentum to the droplet for it to rebound to a greater height and produce satellite droplets -- a phenomenon extensively documented in prior research \cite{aria2014physicochemical,zhang2022impact}. To assess the ability of our model to capture subsequent droplet dynamics, we conducted extended numerical simulations for a droplet with \textit{We}=14.1 (Supplementary Movie 2). As depicted in Figure \ref{fig:jet}, the numerical outcomes effectively replicate the entire sequence of events: secondary impact, jet formation, and eventual satellite droplet generation. These findings underscore the reliability of the proposed contact algorithm and the precision of the cohesion model integrated in this investigation.

\subsection{Sensitivity of the numerical simulation to spatial discretization} 
To assess the influence of spatial resolution on model performance, we compared the temporal evolution of droplet diameter and height from numerical simulations with experimental data under different particle spacings. It was observed that the droplet radius increases as the initial particle spacing decreases. This trend is attributed to a reduction in numerical artificial viscosity with finer spatial resolution; diminished viscous effects facilitate easier deformation of the droplet. As illustrated in Figure \ref{fig:droplet_evolution}(a), the simulated $t_{\max}$ remains approximately 5 ms with increasing spatial resolution. Furthermore, a comparison of the droplet height evolution shows that  $t_{\uparrow}$  stays around 3 ms across all tested particle spacings, as seen in Figure \ref{fig:droplet_evolution}(b). These results indicate consistent agreement with experimental observations and suggest that the model exhibits low sensitivity to the initial particle spacing.  
\begin{figure}[htbp]
\centering
\begin{tikzpicture}[scale=0.6]
    \begin{axis}[ 
        ylabel={Diameter (mm)},
        xlabel={Time (ms)},
        legend pos=  south east,  
        grid=major,
        grid style={dotted,gray!50},
        title={(a) Diameter Evolution},
        xmin=-1, xmax=8,
        ymajorgrids=true,
        xmajorgrids=true
    ]  
    
    \addplot[red, thick, mark=diamond*, mark size=1] table[x=Time, y=Diameter] {data/f07.csv};
    \addlegendentry{$\Delta x_f = 0.07$ mm}

    \addplot[blue, thick, mark=square*, mark size=1] table[x=Time, y=Diameter] {data/f06.csv};
    \addlegendentry{$\Delta x_f = 0.06$ mm}
    
    \addplot[green, thick, mark=diamond*, mark size=1] table[x=Time, y=Diameter] {data/f05.csv};
    \addlegendentry{$\Delta x_f = 0.05$ mm}
    
    \addplot[black!70!black, thick, mark=triangle*, mark size=1] table[x=Time, y=Diameter] {data/f04.csv};
    \addlegendentry{$\Delta x_f = 0.04$ mm} 
    
    \end{axis}
\end{tikzpicture} 
\begin{tikzpicture}[scale=0.6]
    \begin{axis}[ 
        ylabel={Height (mm)},
        xlabel={Time (ms)},
        legend pos= south east,  
        grid=major,
        grid style={dotted,gray!50},
        title={(b) Height Evolution},
        xmin=-1, xmax=8,
        ymajorgrids=true,
        xmajorgrids=true
    ]
    \addplot[red, thick, mark=diamond*, mark size=1] table[x=Time, y=High] {data/f07.csv};
    \addlegendentry{$\Delta x_f = 0.07$ mm}

    \addplot[blue, thick, mark=square*, mark size=1] table[x=Time, y=High] {data/f06.csv};
    \addlegendentry{$\Delta x_f = 0.06$ mm}

    \addplot[green!70!black, thick, mark=triangle*, mark size=1] table[x=Time, y=High] {data/f05.csv};
    \addlegendentry{$\Delta x_f = 0.05$ mm}
    
     \addplot[black, thick, mark=*, mark size=1] table[x=Time, y=High] {data/f04.csv};
    \addlegendentry{$\Delta x_f = 0.04$ mm} 
    
    \end{axis}
\end{tikzpicture} 
\caption{Time evolution of droplet diameter and height for $\textit{We} = 14.1$ under different mesh resolutions ($\Delta x_f$). The results demonstrate the sensitivity of the numerical simulation to spatial discretization.}
\label{fig:droplet_evolution}
\end{figure}

\section{Conclusions}\label{Sec7}
This study presents a rigorously derived nonlocal SPH framework for simulating 3D pancake bouncing dynamics on superhydrophobic microcone arrays. By formulating intermolecular forces to model surface tension through strict theoretical derivations—linking microscopic interactions to the macroscopic surface tension coefficient—we eliminate empirical parameter dependencies. Additionally, the proposed nonlocal contact repulsion mechanism resolves fundamental challenges in defining local normals on irregular microstructures, ensuring numerical stability for high-Weber-number impacts. These advances provide a mathematically grounded perspective for SPH-based modeling of complex droplet-substrate interactions, significantly expanding the method’s capability to capture microstructure-driven phenomena.
Numerical validations demonstrate excellent agreement with experimental data, underscoring the model’s accuracy in predicting droplet dynamics. The framework’s versatility opens avenues for applications requiring precise droplet control, such as anti-icing surfaces and microfluidic systems.
Future work will extend this methodology to dynamic wetting scenarios and multi-droplet interactions, further addressing open questions in interfacial flows while leveraging the nonlocal formulations developed here.


\section*{Funding}
This work is supported by the CAS AMSS-PolyU Joint Laboratory of Applied Mathematics (No. JLFS/P-501/24).
The first author was partially supported by the Hong Kong Research Grants Council RFS grant RFS2021-5S03 and GRF grant 15305624.
The third author was partially supported by the Hong Kong Polytechnic University Postdoctoral Research Fund 1-W30N.

\section*{Declaration of interests}
The authors report no conflict of interest.

\appendix
\section{Proof of Lemma \ref{quadruple_integral}}\label{appendixA}
\begin{proof}
We begin by defining the integral functional:
\[ G_f: = \int_0^{R_0} \int_{0}^{\pi} \int_{0}^{\arccos(s/R_0)} \int_{s/\cos\theta}^{R_0} f(r)\cos\theta \, r^2\sin\theta \, dr d\theta d\phi ds. \]

The integral simplifies in three steps. First, integrating over $\phi$ yields:
\begin{align}
G_f &= \pi \int_0^{R_0} \int_{0}^{\arccos(s/R_0)} \cos\theta\sin\theta \int_{s/\cos\theta}^{R_0} f(r) r^2 dr d\theta ds. \label{simp_G}
\end{align}

Next, applying the substitution $u = s/R_0$ transforms \eqref{simp_G} to:
\[ G_f = \pi R_0 \int_0^1 \int_{0}^{\arccos u} \cos\theta\sin\theta \int_{R_0 u/\cos\theta}^{R_0} f(r) r^2 dr d\theta du. \]

We then exchange the integration order of $u$ and $\theta$:
\[ G_f = \pi R_0 \int_0^{\pi/2} \cos\theta \sin\theta \int_{0}^{\cos\theta} \int_{R_0 u/\cos\theta}^{R_0} f(r) r^2 dr du d\theta. \]

A second exchange of integration order yields:
\begin{align}
G_f &= \pi R_0 \int_0^{\pi/2} \cos\theta\sin\theta \int_{0}^{R_0} f(r) r^2 \int_{0}^{r\cos\theta/R_0} du dr d\theta \notag \\
&= \pi \int_0^{\pi/2} \cos^2\theta\sin\theta d\theta \int_{0}^{R_0} f(r) r^3 dr \notag \\
&= \frac{\pi}{3} \int_{0}^{R_0} f(r) r^3 dr, \label{final_G}
\end{align}
where we have used $\int_0^{\pi/2} \cos^2\theta\sin\theta d\theta = \frac{1}{3}$.

For the special case where $f(|\boldsymbol{x}|) = \partial_{|\boldsymbol{x}|}W(\boldsymbol{x},h)$ with $W$ being a kernel function, we proceed as follows. The normalization condition requires:
\begin{align}\label{eq_1}
1 = \int W(\boldsymbol{x},R_0)d\boldsymbol{x} = 4\pi\int_{0}^{R_0} \tilde{W}(r) r^2 dr,
\end{align}
where $\tilde{W}(|\boldsymbol{x}|) = W(\boldsymbol{x},h)$.

Evaluating \eqref{final_G} through integration by parts:
\begin{align}
G_f &= \frac{\pi}{3} \int_{0}^{R_0} \partial_r\tilde{W}(r) r^3 dr \notag \\
&= \frac{\pi}{3} \left[ \left. r^3 \tilde{W}(r) \right|_{0}^{R_0} - 3\int_{0}^{R_0} \tilde{W}(r) r^2 dr \right] \notag \\
&= -\pi \int_{0}^{R_0} \tilde{W}(r) r^2 dr. \label{intermediate}
\end{align}

Combining \eqref{eq_1} with \eqref{intermediate} establishes:
\[ G_f = -\frac{1}{4}. \]
\end{proof}


\bibliographystyle{siamplain}
\bibliography{data/BIB}

 
\end{document}